\documentclass[sigplan, 10pt, screen]{acmart}
\renewcommand\footnotetextcopyrightpermission[1]{}
\AtBeginDocument{%
  \providecommand\BibTeX{{%
    \normalfont B\kern-0.5em{\scshape i\kern-0.25em b}\kern-0.8em\TeX}}}

\usepackage{pifont} 
\usepackage{amsmath}
\usepackage{array, xspace, xcolor, algorithm}
\usepackage[noend]{algpseudocode}
\usepackage{varwidth}
\usepackage{caption}
\usepackage{listings}
\usepackage[english]{babel}
\usepackage{amsthm} 
\newtheorem{question}{Question}
\newtheorem{lemma}{Lemma}
\newtheorem{theorem}{Theorem}

\newtheorem{definition}{Definition}[]
\usepackage{booktabs} 
\usepackage{subcaption}
\usepackage{adjustbox}      
\usepackage{float}
\usepackage{pifont}
\usepackage{enumitem} 
\usepackage{hhline} 
\usepackage{multirow}

\usepackage{titletoc}
\usepackage{mathtools}

\DeclarePairedDelimiter\floor{\lfloor}{\rfloor}
\usepackage{multicol}
\usepackage[skins]{tcolorbox} 
\tcbuselibrary{breakable}
\usepackage{soul}
\usepackage{mdframed}
\usepackage{changepage} 
\usepackage[rightcaption]{sidecap} 
\usepackage{graphicx}

\newcommand{\algo}{PrestigeBFT\xspace}

\def\wct<#1>{\raisebox{.5pt}{\textcircled{\raisebox{-.9pt} {#1}}}}
\def\pn<#1>{\raisebox{.5pt}{\textcircled{\raisebox{-.5pt} {\scriptsize #1}}}}

\usepackage{tikz,pgfplots,pgfplotstable} 
\usetikzlibrary{decorations.pathreplacing,calc,shapes,positioning,tikzmark} 
\usetikzlibrary{patterns}

\mathchardef\mhyphen="2D
\newcommand{\RN}[1]{
  \textup{\uppercase\expandafter{\romannumeral#1}}
}

\acmConference[Work in Progress]{}{Technical Report}{2023}
  
\pgfplotsset{compat=1.18}
\begin{document}

\title{\algo: Revolutionizing View Changes in BFT Consensus Algorithms with Reputation Mechanisms}

\author{Gengrui Zhang, Fei Pan, Sofia Tijanic, and Hans-Arno Jacobsen}
\affiliation{
  \institution{\textit{University of Toronto}}
  \country{}}
\email{{gengrui.zhang, fei.pan, sofia.tijanic}@mail.utoronto.ca, jacobsen@eecg.toronto.edu}

\renewcommand{\shortauthors}{Gengrui Zhang, et al.}

\begin{abstract}
This paper proposes \algo, a novel leader-based BFT consensus algorithm that addresses the weaknesses of passive view-change protocols. Passive protocols blindly rotate leadership among servers on a predefined schedule, potentially selecting unavailable or slow servers as leaders. \algo proposes an active view-change protocol using reputation mechanisms that calculate a server's potential correctness based on historic behavior. The active protocol enables servers to campaign for leadership by performing reputation-associated work. As such, up-to-date and correct servers with good reputations are more likely to be elected as leaders as they perform less work, whereas faulty servers with bad reputations are suppressed from becoming leaders by being required to perform more work. Under normal operation, \algo achieves $5\times$ higher throughput than the baseline that uses passive view-change protocols. In addition, \algo remains unaffected under benign faults and experiences only a $24\%$ drop in throughput under a variety of Byzantine faults, while the baseline throughput drops by $62\%$ and $69\%$, respectively.
\end{abstract}

\maketitle

\section{Introduction}
\label{sec:introduction}
The rapid development of distributed systems has spurred extensive research on Byzantine fault-tolerant (BFT) consensus algorithms. Among them, leader-based BFT algorithms have been favored by practical applications due to their high performance. These algorithms operate state machine replication to produce deterministic results using two distinct protocols: the view-change and replication protocol. The view-change protocol selects a leader (primary) for each view, while the replication protocol enables the leader to initiate consensus with followers (backups). While prior research has primarily concentrated on the replication protocol to optimize system performance, the importance of the view-change protocol has often been overlooked. However, today, the view-change protocol has become a critical factor in system performance as view changes occur more frequently than previously assumed due to various factors associated with fault tolerance, performance criteria, and decentralization with fairness imperatives~\cite{causesofcrash, guo2013failure, amir2008byzantine, clement2009making, aublin2013rbft, kokoris2018omniledger, diemConsensus}. 
For example, \wct<1> network problems, bursty workloads, operator errors, and software bugs can result in leader failure\cite{causesofcrash}, resulting in more frequent view changes as applications scale up~\cite{guo2013failure}.
\wct<2> Since faulty leaders can intentionally slow down processing without triggering timeouts~\cite{amir2008byzantine}, some approaches monitor a leader's performance and invoke a view change if the performance falls below a set threshold~\cite{clement2009making, aublin2013rbft}.
Furthermore, \wct<3> faulty leaders can unfairly handle client requests~\cite{kokoris2018omniledger, zhang2020byzantine}, resulting in some approaches frequently changing leadership to mitigate unfairness~\cite{diemConsensus}. Evidently, the view-change protocol is vital for system performance, especially in blockchain applications where frequent view changes are becoming the norm~\cite{cachin2017blockchain, diem}.

\begin{SCfigure}[1][t]
    \centering
    \includegraphics[width=0.45\linewidth]{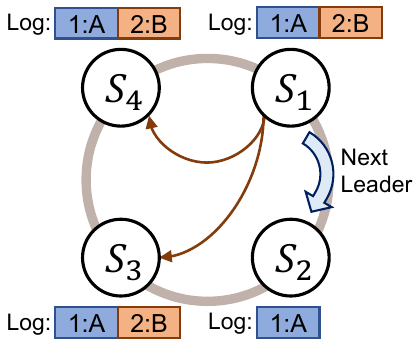}
    \caption{The passive view-change protocol follows a fixed schedule to rotate leadership; it cannot avoid an already crashed server or a slow server to become a leader ($S_2$).}
    \label{fig:passive-rotation}
\end{SCfigure}

Numerous state-of-the-art BFT algorithms (e.g., \cite{yin2019hotstuff, gueta2019sbft, stathakopoulou2022state, stathakopoulou2019mir, danezis2022narwhal, spiegelman2022bullshark, keidar2021all}), despite optimizing the replication protocol extensively, all rely on similar view-change mechanisms introduced by PBFT~\cite{castro1999practical}. The passive protocol follows a predefined schedule to rotate leadership among servers: for a total of $n$ servers, a leader server ($L$) for a view $V$ is decided such that $L$ = V mod $n$. For example, in Figure~\ref{fig:passive-rotation}, $S_1$ is the leader in $V1$; when a view change occurs, $S_2$ becomes the leader in $V2$, and $S_3$ in $V3$, and so on.

Unfortunately, the passive protocol lacks robustness and efficiency. During a view change, since all servers blindly follow a predefined schedule and rotate leadership, the passive protocol cannot skip a scheduled server that is already unavailable, which leads to weak robustness. In addition, the passive protocol can result in inefficiency in replication because it cannot ensure \emph{optimistic responsiveness} (OR)~\cite{pass2018thunderella}, which requires a non-faulty leader to be up-to-date and make immediate progress after being assigned~\cite{attiya1994bounds, yin2019hotstuff}. However, when a slow server is rotated to be the leader by a passive protocol, it must sync to become up-to-date first and then starts operating consensus. Thus, BFT algorithms using the passive view-change protocol have to add a sync-up phase after each value is committed to obtain OR (e.g., from two-phase to three-phase in HotStuff~\cite{yin2019hotstuff}), but this comes at the cost of reduced throughput and increased latency due to additional messages and rounds.

Figure~\ref{fig:passive-rotation} shows an $n=4$ system where $S_1$ is the leader in $V1$, and value $B$ has been replicated at sequence number (SN) $2$ among $S_1$, $S_3$, and $S_4$. Next, we assume that $S_1$ fails, causing a view change to take place. \wct<1> If $S_2$ has already crashed, $S_2$ will still blindly be assigned to be the leader in $V2$, and the system must wait for timeouts from $f+1$ servers to realize that $S_2$ has failed before moving on to $S_3$ in $V3$. On the other hand, \wct<2> if $S_2$ is alive, it cannot make immediate progress because it must first sync to become up-to-date (i.e., know the highest SN). Therefore, in a circle of leadership rotations among $n{=}3f{+}1$ servers (i.e., from $S_1$ to $S_n$), the probability of encountering an unavailable or slow leader is $f/(3f+1) \approx 33\%$ in the worst case. 

In order to improve the robustness and efficiency of view changes, we set out to investigate active view changes where servers no longer follow a predefined schedule. Raft's leader election mechanism ushered a way for designing active view-change protocols under non-Byzantine (benign) failures~\cite{ongaro2014search}. It allows servers to actively campaign for leadership upon detecting a leader's failure and vote for an up-to-date server to become a new leader. Consequently, it can prevent unavailable and slow servers from becoming leaders. However, in the context of BFT, this approach alone is insufficient. While servers are empowered to campaign for leadership, it also opens the door for Byzantine servers to repeatedly initiate new view changes to seize control of leadership and neglect replication (repeated view change attacks). Therefore, resolving this issue is essential for active view changes to be deployed under BFT.

To tackle this challenge, we propose \algo, a new BFT consensus algorithm with an active view-change protocol featuring reputation mechanisms. \textbf{Our reputation mechanism utilizes a server's behavior history to generate a reputation value that reflects the likelihood of the server's correctness. This reputation value then determines the probability of the server being selected as a new leader during the active view-change protocol.}

\algo penalizes suspiciously faulty behavior with worsening reputations, while it rewards protocol-obedient behavior with improving reputations. A server's reputation value is utilized to assess its likelihood of becoming a new leader. During view changes, the active view-change protocol imposes computational work on each leadership campaigner, where the difficulty of the computation is determined by the campaigner's reputation value. Correct servers, who exhibit protocol-obedient behavior and thus maintain a "good" reputation, perform negligible computational work. On the other hand, faulty servers, whose behavior history has led to a "bad" reputation, perform more time-consuming computational work. Thus, correct and up-do-date servers are more likely to be elected than faulty servers over time.

Equipped with its reputation-embedded active view-change protocol, \algo demonstrates both robustness and efficiency. During normal operation, its replication achieves optimistic responsiveness with $5\times$ higher throughput than HotStuff~\cite{yin2019hotstuff}. Under benign faults, passive view-change protocols suffer from an approximate $65\%$ drop in throughput, whereas \algo's performance remains unaffected. Furthermore, under a variety of Byzantine faults, \algo's reputation engine swiftly suppresses faulty servers from attaining leadership, resulting in only a $24\%$ drop in throughput compared to that under normal operation. 

\algo makes the following key contributions:

\begin{itemize}
   
    \item Its view-change protocol is the first active protocol operating under BFT. By enabling servers to proactively campaign for leadership, it prevents the election of unavailable or slow servers, thereby achieving optimistic responsiveness.

    \item Its reputation mechanisms effectively convert a server's behavior history during replication and view changes into a reputation value that indicates the server's likelihood of being correct. The reputation value is crucial in determining a server's eligibility for leadership in the view-change protocol.
    
    \item It demonstrates a unique combination of robustness and efficiency, with improved performance even under Byzantine failures. Faulty servers are quickly suppressed during view changes, and their probability of being elected rapidly decreases after they perform attacks that relegate their reputations.
\end{itemize}

\begin{figure*}[t]
\minipage{0.33\textwidth}
    \includegraphics[width=\textwidth]{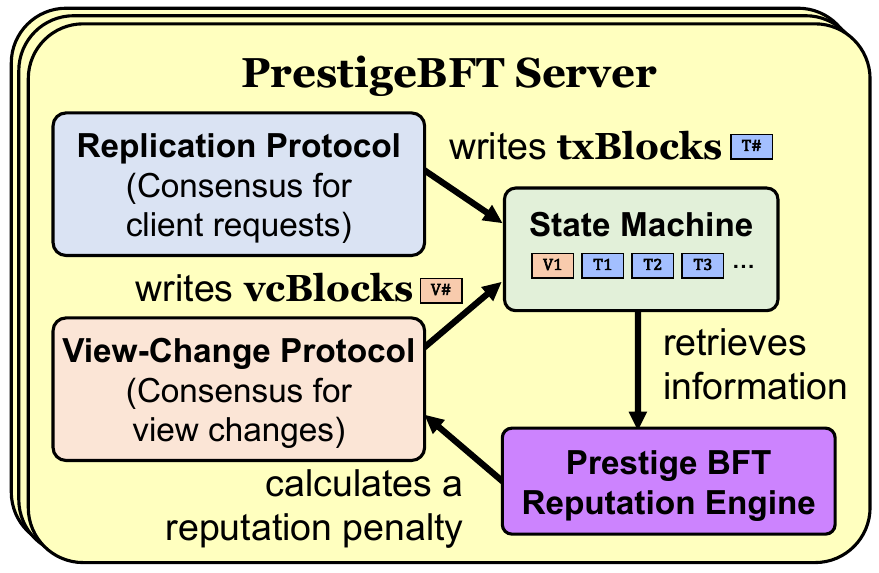}
    \caption{\algo architecture.}
    \label{fig:architecture}
\endminipage \hfill
\minipage{0.66\textwidth}
    \includegraphics[width=\textwidth]{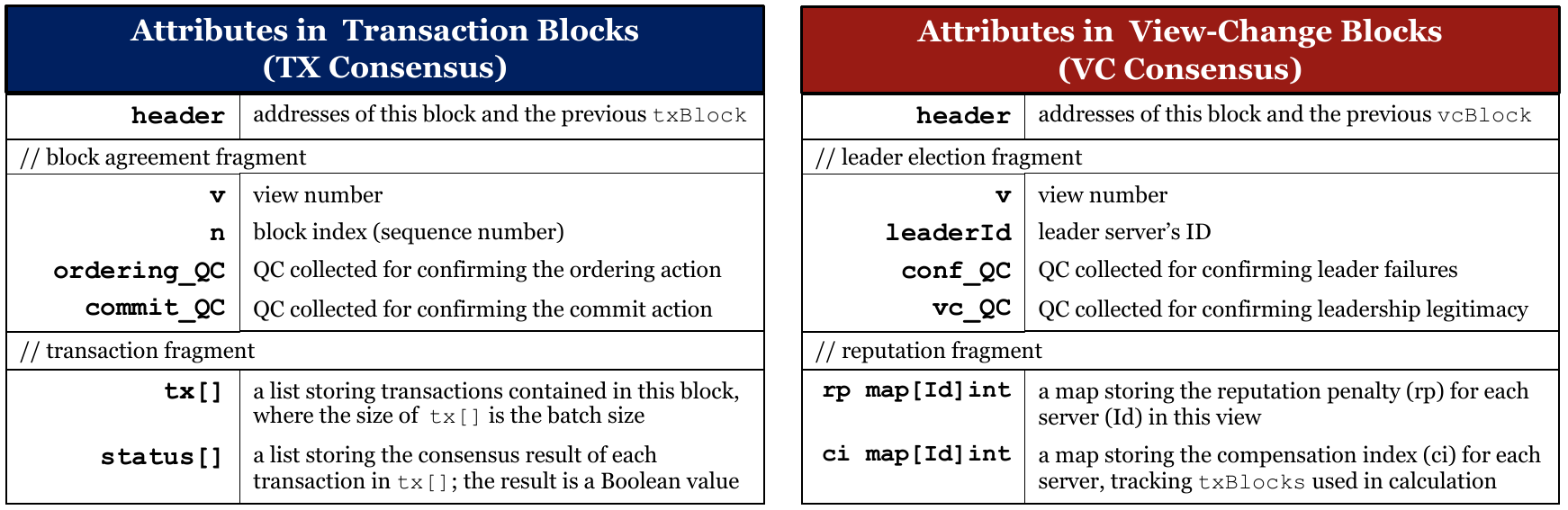}
    \caption{Attributes in \texttt{vcBlock}s and \texttt{txBlock}s.}
    \label{fig:blocks}
\endminipage 
\end{figure*}

\section{\algo overview}
\label{sec:overview}
\algo, similar to other state-of-the-art leader-based BFT consensus algorithms (e.g., PBFT~\cite{castro1999practical} and HotStuff~\cite{yin2019hotstuff}), moves through a succession of system configurations called views. 
Views are integers that increase monotonically. Each view starts with a \emph{view-change period} conducted by the view-change protocol that decides a leader and may follow with a \emph{replication period} conducted by the replication protocol that achieves consensus for client requests.

\textbf{\algo architecture.} 
Besides the view-change and replication protocols, \algo establishes a unique reputation mechanism (shown in Figure~\ref{fig:architecture}). At any given time, a server operates in one of four states: \textit{follower}, \textit{redeemer}, \textit{candidate}, and \textit{leader}. Under normal operation, there is one leader and the other servers are followers, where all servers operate under the replication protocol, which conducts consensus for committing client requests, producing \texttt{txBlock}s (shown in Figure~\ref{fig:blocks}) that record quorum certificates ($QC$s)~\cite{gifford1979weighted}. When a view change is invoked, the view-change protocol produces \texttt{vcBlock}s that record leadership and servers' reputation information in the new view. Note that \texttt{txBlock}s and \texttt{vcBlock}s are the deterministic consensus results of replication and view change.

Each server has a reputation engine that has predefined rules to calculate reputation values. The reputation engine is utilized during view changes, wherein it retrieves the necessary states of \texttt{txBlock}s and \texttt{vcBlock}s stored in the state machine through read operations. Based on the retrieved information, it calculates a \textbf{reputation penalty} (an integer). The penalty reflects the server's likelihood of correctness and determines the work that the server performs to become a candidate. The process for a server, initially as a follower, to become a leader works at a high level as follows:

\begin{enumerate}
    \item Each server is initially a follower. If a follower triggers a view change confirmed by $f{+}1$ servers in view $V$, it will campaign for leadership and become a redeemer.

    \item The redeemer increments its view to $V+1$ and gets its reputation penalty from its reputation engine. It then performs computation determined by the reputation penalty; once completed, it transitions to a candidate.

    \item The candidate starts a leader election by collecting votes from $2f+1$ servers; if it succeeds in time, it becomes the leader in view $V+1$.

    \item The leader prepares a \texttt{vcBlock} including the election result and updated reputation information and then broadcasts the \texttt{vcBlock} to others.
\end{enumerate}

The reputation penalty is critical in successful elections. Servers with higher penalties are more suspected of being faulty and will perform more computational work, making the election process substantially more difficult for them. Next, we show how the reputation mechanism translates a server's past behavior to a reputation penalty that serves as an indicator of the server's correctness.

\section{\algo's reputation mechanism}
\label{sec:reputation}

The reputation mechanism analyzes a server's behavior in past replication and view changes and produces a reputation penalty ($rp$), represented as an integer. A higher penalty corresponds to a worse reputation and indicates a higher level of suspicion that the server may be malicious.

The calculation of $rp$ involves two steps: penalization and compensation. It increases $rp$ for a server that seeks to become the leader (penalization) and reduces $rp$ if the server exhibits good historic behavior (compensation). \textbf{Therefore, a server's reputation penalty may either increase, decrease, or remain unchanged from its current value after being assessed by the reputation engine.} Algo.~\textsc{CalcRP} shows the calculation workflow.

{\footnotesize
\setlength{\intextsep}{0pt}
\begin{algorithm}[t]
\caption{\textsc{Calculate-Reputation-Penalty (CalcRP)}}
\label{algo:calcrp}
\begin{algorithmic}[1]
    \Require $V'$, vcBlock, txBlock, Id \Comment{$V'$ is the new view}
    \State $V$, $rp$ = \texttt{vcBlock}.v, \texttt{vcBlock}.rp[Id]
    \Comment{get current view and penalty}
    \State new\_rp\_temp = Eq.\ref{eq:increase-penalty}($V'$, $V$, rp)

    \State ti, ci = \texttt{txBlock}.n, \texttt{vcBlock}.ci[Id] \label{re:ci} \Comment{get ti and ci from stored blocks}
    \State P = $[]$int$\lbrace rp \rbrace$ \Comment{init penalty set with $rp$}
    \While{vcBlock.header.preVcBlock != nil} \label{re:pstart}
        \State vcBlock = vcBlock.header.preVcBlock \Comment{iterates to the first block}
        \State P.add(vcBlock.rp[Id]) \label{re:pend}
        \Comment{P contains all past penalties}
    \EndWhile
    \State delta\_tx, delat\_vc = Eq.\ref{eq:delta-tx}$(ti, ci)$, Eq.\ref{eq:delta-vc}$(P, rp)$
    \State new\_rp = Eq.\ref{eq:decrease-penalty}(new\_rp\_temp, delta\_tx, delat\_vc)
    \State \Return new\_rp, ci
\end{algorithmic}
\end{algorithm}
}

\subsubsection*{Init}
The initial view and penalty can be defined differently. For simplicity, we set initial view to $V1$ and $rp^{(1)}=1$.

\subsubsection*{Step 1: Penalization}
A server's $rp$ is increased for campaigning for leadership for the new view ($V'$) following Eq.~\ref{eq:increase-penalty}. The increase in $rp$ is the increase in view numbers, which prevents Byzantine servers from overloading the view data structure (an integer). The higher the increase in views, the higher the penalty will be. Correct servers will always increase their views by one.

\begin{equation}
\label{eq:increase-penalty}
    rp^{(V')}_{temp} = rp^{(V)} + (V' - V)
\end{equation} 
After applying penalization, the calculation proceeds to the second step, which involves deducting from the increased $rp$ if the server's behavior history meets certain criteria.

\begin{figure*}[t]
\begin{subfigure}[b]{0.3\textwidth}
    \centering
    \includegraphics[width=0.97\linewidth]{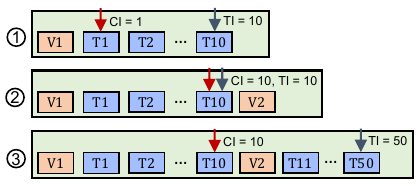}
    \caption{Examples of calculating incremental log responsiveness ($\delta_{tx}$).}
    \label{fig:compesationindex}
\end{subfigure}
\hfill
\begin{subfigure}[b]{0.35\textwidth}
    \centering
    \includegraphics[width=\linewidth]{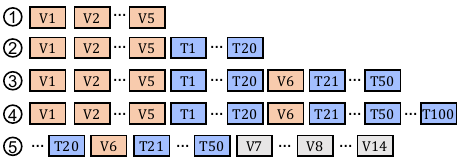}
    \caption{Examples of historic behavior. The server is the leader in highlighted \texttt{V\#} but not in gray \texttt{V\#}.}
    \label{fig:reputation-behavior}
\end{subfigure}
\hfill
\begin{subfigure}[b]{0.33\textwidth}
    \centering
    \begin{adjustbox}{width=0.99\linewidth}
    \begin{tabular}{c|cc >{\bfseries}c|lcc >{\bfseries}c}
         \hline
         & $\texttt{ci}$ & $\texttt{ti}$ & $\delta_{tx}$ & \multicolumn{1}{c}{$\mathcal{P}$} & $\delta_{vc}$ & $\delta$ & $rp^{(V')}$\\
         \hline
         \wct<1> & 1 & 1 & 0 & $\{1,2,3,4,5\}$ & 0.19 & 0 & 6 \\
         \wct<2> & 1 & 20 & 1 & $\{1,2,3,4,5\}$ & 0.19 & 1.14 & 5 \\
         \hline
         \wct<3> & 20 & 50 & 0.6 & $\{1,2,3,4,5,5\}$ & 0.25 & 0.89 & 6 \\
         \wct<4> & 20 & 100 & 0.8 & $\{1,2,3,4,5,5\}$ & 0.25 & 1.2 & 5 \\
         \hline
         \wct<5> & 20 & 50 & 0.6 & $\mathcal{P}^5$ & 0.36 & 1.29 & 5 \\
         \hline
    \end{tabular}
    \end{adjustbox}
    {\small \newline
    $\mathcal{P}^{5} {=} \lbrace 1, 2, 3, 4, 5, 5, ... , 5 \rbrace$
    //$5$ appears 10 times
    }
    \caption{Breakdown of calculating compensation in~\ref{fig:reputation-behavior}; $rp^{(V')}$ will be the new penalty in $V'$.}
    \label{fig:reputation-calculations}
\end{subfigure}
\caption{Examples of calculating reputation penalty ($rp$), where ``\texttt{V\#}'' are \texttt{vcBlock}s and ``\texttt{T\#}'' are \texttt{txBlock}s.}
\label{fig:reputation-examples}
\end{figure*}

\noindent\begin{minipage}{.33\linewidth}
\begin{equation}
  \label{eq:delta-tx}
    \delta_{tx} =  \dfrac{\texttt{ti} - \texttt{ci}}{\texttt{ti}}
\end{equation}
\end{minipage}
\hfill
\begin{minipage}{.66\linewidth}
  \begin{equation}
  \label{eq:delta-vc}
    \delta_{vc} = 1- Sigmoid(\dfrac{rp^{(V)}-\mu_{\mathcal{P}}}{\sigma_{\mathcal{P}}})
  \end{equation}
\end{minipage}

\begin{equation}
    \label{eq:decrease-penalty}
    \delta = rp^{(V')}_{temp} C_{\delta} \delta_{tx} \delta_{vc},
    \quad rp^{(V')} = rp^{(V')}_{temp} - \floor{\delta} 
\end{equation}

\subsubsection*{Step 2: Compensating good behavior history}
The compensation has two criteria: \emph{incremental log responsiveness} ($\delta_{tx}$ in Eq.~\ref{eq:delta-tx}), which is calculated using \texttt{txBlock}s, and \emph{leadership zealousness} ($\delta_{vc}$ in Eq.~\ref{eq:delta-vc}), which is calculated using \texttt{vcBlock}s. The final $rp$ is calculated in Eq.~\ref{eq:decrease-penalty} with a possible deduction considering both criteria. 

\textbf{The first criterion ($\delta_{tx}$) considers good behavior to be up-to-date replication}. In Eq.~\ref{eq:delta-tx}, \texttt{ti} is the number of \texttt{txBlock}s this server has committed, which is the sequence number of its latest $\texttt{txBlock}$, and \texttt{ci} is the compensation index representing the number of \texttt{txBlock}s this server has used for past compensation stored in the current \texttt{vcBlock} (Line~\ref{re:ci}). Initially, \texttt{ti}=1 and \texttt{ci}=1, so $0 \leq \delta_{tx} \leq 1$.

Figure~\ref{fig:compesationindex} shows three examples of calculating $\delta_{tx}$ for a server (say $S_a$). \wct<1> $S_a$ has replicated $10$ \texttt{txBlock}s in $V1$, so its \texttt{ci}=1 and \texttt{ti}=10. 
\wct<2> If $S_a$ campaigns for leadership for $V2$, after applying Eq.~\ref{eq:delta-tx}, its $\delta_{tx}=0.9$ with $10$ \texttt{txBlock}s used for compensation. If $S_a$ is elected, \texttt{ci}=10. 
\wct<3> If $S_a$ then replicates $50$ \texttt{txBlock}s in total and campaigns for $V3$, its $\delta_{tx}=0.8$. 
Therefore, to receive a higher $\delta_{tx}$, the reputation mechanism entices servers to replicate more \texttt{txBlock}s.

Considering log responsiveness is a common aspect of reputation-based approaches (e.g., DiemBFT~\cite{diemConsensus}), where the more transactions a server replicates, the more reliable it is perceived to be. Nevertheless, since \algo is an active view-change protocol, relying solely on behavior in replication is inadequate. \emph{It is crucial to consider cases where Byzantine servers repeatedly acquire leadership but make limited progress in replication.} Therefore, in addition to log responsiveness, our reputation mechanism considers a server's historic penalties in previous view changes as well. 

\textbf{The second criterion ($\delta_{vc}$) considers good behavior to be having gradually increasing penalties during past view changes.} 
$\delta_{vc}$ first computes the z-score of a server's current penalty in relation to its past penalties, taking into account the rate of change of the current penalty over past penalties. It retrieves the server's historic penalties stored in \texttt{vcBlock}s and adds them into a set $\mathcal{P}$ (Line~\ref{re:pstart} to~\ref{re:pend}); then, it calculates the mean ($\mu_{\mathcal{P}}$) and standard deviation ($\sigma_{\mathcal{P}}$) of $\mathcal{P}$. Finally, the \texttt{Sigmoid} function normalizes the z-score between $0$ and $1$. Thus, $0 < \delta_{vc} < 1$.

Therefore, a higher $\delta_{vc}$ value indicates a slower increase in penalties, which is more towards the behavior of correct servers. Since correct servers adhere to the protocol for triggering view changes, they are unlikely to be penalized while regaining leadership in each view; doing so would require proactively performing significantly increasing computation to ``fight against'' the penalty increase (discussed in \S\ref{sec:algo:active-vc}).

Finally, after obtaining $\delta_{tx}$ and $\delta_{vc}$, the new $rp$ is calculated by Eq.~\ref{eq:decrease-penalty}, where $C_{\delta}$ is a constant that may be used by different applications to adjust the effect of $\delta_{tx}$ and $\delta_{vc}$. For simplicity, we set $C_{\delta}=1$. Since $0 {\leq} \delta_{tx} {\leq} 1$ and $0 {<} \delta_{vc} {<} 1$, the deduction ($\delta$) is a portion of the increased penalty of Eq.~\ref{eq:increase-penalty}; i.e., 
$0 \leq \delta< rp^{(V')}_{temp}$.
As such, the reputation mechanism will raise suspicion of  malicious behavior by increasing $rp$ if it observes a pattern of penalized leadership repossession with limited replication. Conversely, it will decrease $rp$ if it observes historic behavior that implies high log responsiveness and gradually increased or unchanged historic penalties.

\textbf{More examples.}
Figure~\ref{fig:reputation-behavior} and~\ref{fig:reputation-calculations} provide five examples of calculating $rp$ for different server behavior. They show how the reputation mechanism responds to suspicious malicious behavior and to protocol-obedient behavior.

\begin{enumerate}
    \item[\wct<1>] Server $S_a$ has been the leader from $V1$ to $V5$ without replication. Thus, its $\delta_{tx}$ remains $0$, resulting in no compensation in Eq.~\ref{eq:decrease-penalty} with its $rp$ only increasing. If $S_a$ campaigns for leadership for the next view ($V6$), its $rp$ will increase to~$rp^{(6)} {=} 6$.
    
    \item[\wct<2>] If $S_a$ conducts consensus for $20$ \texttt{txBlock}s in $V5$ and then campaigns for leadership for the next view, its ${\delta_{tx} {=} 1}$. It will receive a compensation of $1$ with unchanged $rp$ ($rp^{(6)} {=} rp^{(5)} {=} 5$).
\end{enumerate}

\noindent \textit{Analysis.}
$S_a$'s behavior in \wct<1> is extremely suspicious to be malicious, as it keeps repossessing leadership without making progress in replication. The reputation mechanism captures this pattern and keeps increasing its reputation penalty. Compared to~\wct<1>, $S_a$'s behavior in \wct<2> reduces suspicion as it starts to replicate transactions; the reputation mechanism encourages this behavior and grants compensation.

\begin{enumerate}
    \item[\wct<3>] In $V6$, $S_a$ replicates $20$ more \texttt{txBlock}s. If $S_a$ starts to campaign for $V7$, its ${\delta_{tx}{=}0.6}$ by Eq.~\ref{eq:delta-tx} as \texttt{ci}=20 and \texttt{ti}=50. Thus, $S_a$ receives no compensation with its penalty increasing to $rp^{(7)}=6$.
    
    \item[\wct<4>] If $S_a$ replicates more \texttt{TxBlock}s for a total of $100$ in $V6$, its ${\delta_{tx}=0.8}$. In this case, if $S_a$ campaigns for leadership in $V7$, it will receive a compensation of $1$ with its $rp$ remaining unchanged ($rp^{(7)} {=} 5$).
\end{enumerate}

\noindent \textit{Analysis.}
Incremental log responsiveness ($\delta_{tx}$) expects an increasing number of \texttt{txBlock}s after each compensation (e.g., $S_a$ cannot get compensated in \wct<3> but can be compensated in \wct<4>); this prevents a server from frequently occupying leadership but only making limited progress in replication.

\begin{enumerate}
    \item[\wct<5>] Continuing in \wct<3>, assume $S_a$ is no longer eager to leadership after $V6$, staying as a follower from $V7$ to $V14$ (the gray \texttt{vcBlock}s where other servers are leaders). Its $\delta_{vc}$ kept increasing as its penalty remains unchanged as $5$ from $V7$ to $V14$. If $S_a$ campaigns for leadership in $V15$, it will be compensated by $1$.
\end{enumerate}

\noindent \textit{Analysis.}
$S_a$ in \wct<5> can be compensated as \wct<4> with the same level of replication in \wct<3> when it stops repossessing leadership with a suspicious history. It shows that $\delta_{vc}$ incentivizes servers with a history of increasing penalties to become indifferent to leadership.

The examples show the effectiveness, efficiency, and simplicity of the reputation mechanism in transforming a server's prior actions into an integer value ($rp$) that is indicative of its correctness. In addition, the calculation schema is highly adaptable and can be customized for specific use cases. For example, users can define the criteria for useful \texttt{txBlock}s in $\delta_{tx}$ and alter the impact of $\delta_{tx} \delta_{vc}$ by modifying $C_{\delta}$. Step-by-step calculations are provided in Appendix~\ref{sec:ap:examples}.

\textbf{Features.} 
\textit{The reputation mechanism does not incur additional cost on replication}. The reputation engine is independent (shown in Figure~\ref{fig:architecture}) and involved only when a server has become a redeemer in a view change; i.e., the replication has already stopped.

Moreover, \textit{the reputation mechanism never writes to the state machine}. The reputation engine operates as a ``consultant'' who calculates an $rp$ when called. It never changes a server's $rp$ in the current view; i.e., in a given view $V$, a server's $rp$ remains unchanged throughout $V$. The calculated $rp$ will become a server's new $rp$ in the next view only if it is elected as the new leader through VC consensus. 

\section{The \algo consensus algorithm}
\label{sec:algo}
This section introduces \algo's system model, active view-change protocol, and replication protocol.

\subsection{System model}
\label{sec:algo:system-model}
We adopt the partial synchrony network model introduced by Dwork et al.~\cite{dwork1988consensus}, where there is a known bound $\Delta$ and an unknown Global Stabilization Time (GST), such that after GST, all transmissions between two non-faulty servers arrive within time $\Delta$. \algo does not require network synchrony to provide safety, but it requires partial synchrony to provide liveness.

We use a Byzantine failure model, meaning that faulty servers may behave arbitrarily. We assume independent server failures where each server represents an independent entity. \algo tolerates up to $f$ Byzantine servers out of $n{=}3f+1$ servers (i.e., $f=\floor{\frac{n-1}{3}}$) with no limit on the number of faulty clients. Similar to other BFT algorithms~\cite{yin2019hotstuff, castro1999practical, gueta2019sbft, kotla2007zyzzyva, zhang2020byzantine, danezis2022narwhal}, we do not consider DDoS attacks (e.g., buffer overflow attacks), which are often handled by lower-level mechanisms, such as rate limiting and admission control, outside of the scope of consensus algorithms. 

\textbf{Cryptographic primitives and quorums.} 
\algo applies ($t$, $n$) threshold signatures where $t$ out of $n$ servers collectively sign a message~\cite{shoup2000practical}. Threshold signatures can convert $t$ individually signed messages (of size $O(n)$) into one fully signed message (of size $O(1)$), which can then be verified by all $n$ servers, proving that $t$ servers have signed it. \algo uses threshold signatures to form quorum certificates ($QC$s) in consensus by setting $t{=}2f+1$. 

\textbf{Attack vector and threat model.}
We allow for a faulty server to collude with the other $f{-}1$ faulty servers as well as an unlimited number of faulty clients. The set of faulty servers can change dynamically, with servers transitioning between correct and faulty states, as long as the total number of faulty servers does not exceed $f$. Faulty servers can behave arbitrarily and maliciously, but we assume that faulty servers cannot intervene to prevent the state changes of non-faulty servers specified by the algorithm, which is a fundamental assumption of Byzantine fault tolerance~\cite{attiya2004distributed}. For example, faulty servers cannot prevent a non-faulty server from delivering messages to other non-faulty servers after GST. We also assume that a faulty server (and its colluding faulty servers) are computationally bound. They cannot produce a valid signature of a non-faulty server. It is worth noting that our system model is the same as other standard partially synchronous BFT algorithms~\cite{yin2019hotstuff, buchman2016tendermint, castro1999practical, gueta2019sbft}, without additional assumptions.

\subsection{The active view-change protocol}
\label{sec:algo:active-vc}
The view-change (VC) protocol achieves VC consensus for deciding on a new leader and updating its reputation penalty and compensation index (shown in Figure~\ref{fig:blocks}). The VC protocol attains the following properties:

\begin{enumerate} [label=\textbf{P\arabic*}]
    \item At most one leader can be elected in a view. \label{vc:p1}

    \item An elected leader has the most up-to-date replication (ensuring optimistic responsiveness). \label{vc:property:or}
        
    \item An elected leader's reputation penalty and the correspondingly performed computation can be verified by all non-faulty servers. \label{vc:property:verify}
\end{enumerate}

\begin{figure}
    \centering
    \includegraphics[width=\linewidth]{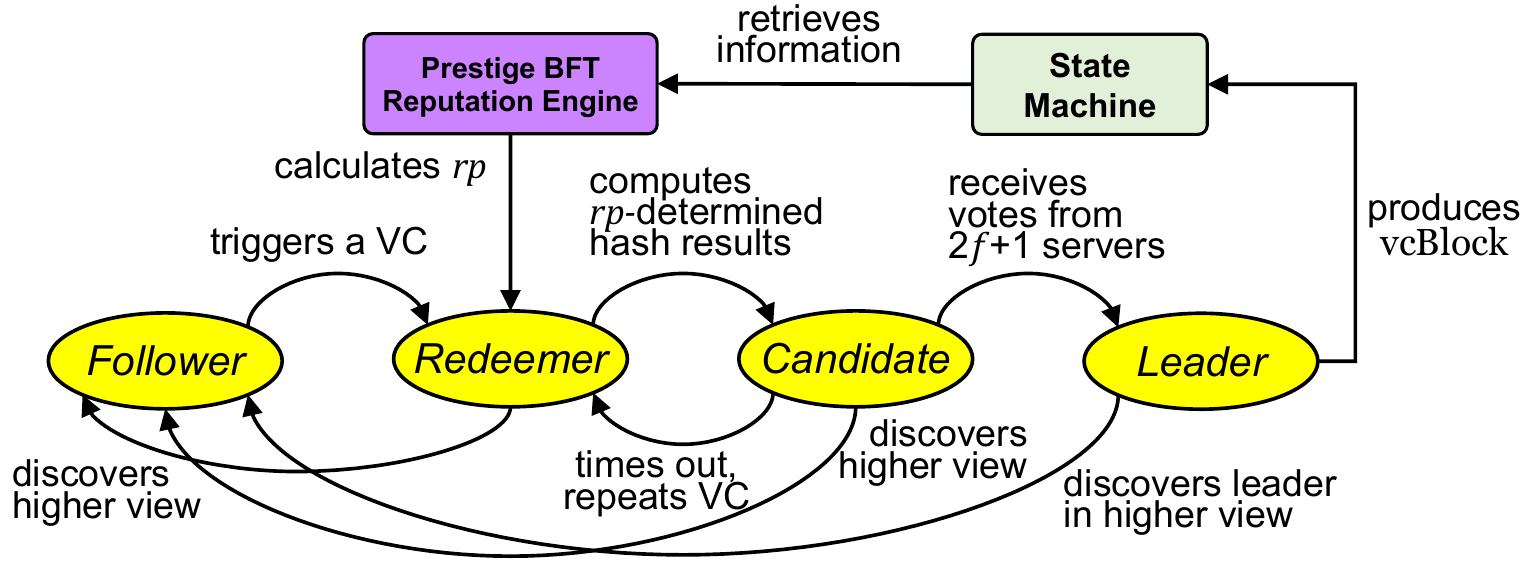}
    \caption{Server states transition in view changes.}
    \label{fig:state-trans}
    \vspace{-1em}
\end{figure}

Next, we describe the VC protocol with server state transitions (illustrated in Figure~\ref{fig:state-trans}) in Algo.~\textsc{State-Transition} and demonstrate how these properties are achieved.

\begin{algorithm}[H]
\algrenewcommand\algorithmicfor{\textbf{upon}}
\footnotesize
\caption{State-Transition}
\begin{algorithmic}[1]
  \Statex \textcolor{blue}{$\triangledown$ As a follower}
  \For{receiving a \textsc{Compt}} \label{aglo:af:procedurecompt}
    \If{\textsc{Compt}.$\delta_c$ is valid} 
    \textbf{send} \textsc{Compt} to \texttt{Leader} \label{algo:af:relay}
    \EndIf
    \State \textbf{new} timer for \textsc{Compt}.tx
    \While{timer}
        \If{\textsc{Compt}.tx is committed} 
        \Return \label{fol:compt-committed} \Comment{Leader is correct}
        \EndIf
    \EndWhile
    \State \textbf{broadcast} \textsc{ConfVC} $\leftarrow$ $\langle V, \textsc{Compt}, \sigma_{S_i} \rangle$ and \textbf{reset} timer \label{algo:af:inspection}
    \While{timer}
        \For{receiving $f+1$ \textsc{ReVC} messages} \label{algo:af:revc}
            \State \textbf{convert} \textsc{ReVC}s to 
            \texttt{conf\_QC} and \textbf{transition} to \texttt{Redeemer} \label{algo:af:confqc}            
            \State \Return \Comment{Leader is faulty; start a new VC}
        \EndFor
    \EndWhile
    \State \textbf{tag client} \textsc{Compt}.c \Comment{Client can be faulty}
  \EndFor

  \For{receiving a \textsc{ConfVC}} \label{algo:af:confvc}
    \If{\textsc{ConfVC.Compt} has been received from \textsc{Compt}.c}
    \State \textbf{send} \textsc{ReVC} $\leftarrow$ $\langle V, \sigma_{S_i} \rangle$
    \EndIf
  \EndFor
 
  \For{receiving a \textsc{CampVC}} \label{algo:af:pro3} \label{algo:fol:pro3} \Comment{voting criteria in \S\ref{sec:ascandidate}}

    \If{\textsc{CampVC}.V' < myVcBlock.v} \Return \EndIf
    \If{!\ref{criterion:1}} \Return \Comment{vote only once in a view} \EndIf 

    \If{!\ref{criterion:2}} \Return \Comment{verify if this VC is confirmed} \EndIf
    \If{\textsc{CampVC}.V > myVcBlock.v}
        \State \Call{SyncUp}{myVcBlock, \textsc{CampVC}.V} 
        \Comment{sync-up view changes} \label{algo:af:syncvc}
    \EndIf

    \State $ti = \textsc{CampVC}.txBlock.n$
    \If{$ti$ < myTxBlock.n} \Return \Comment{\ref{criterion:3} verifying replication} \label{algo:candy:replistatus} \EndIf
    
    \If{$ti$ > myTxBlock.n}
        \State \Call{SyncUp}{myTxBlock, $ti$} \Comment{sync-up replication} \label{algo:af:synctx}
    \EndIf
     
    \State $rp', ci'$ = \Call{CalcRP}{\textsc{CampVC}.V', myVcBlock, txBlock, CandID} \label{algo:af:rcalrp}

    \If{$ci' != \textsc{CampVC}.ci$} \Return \label{algo:af:verici} \EndIf

    \If{$rp' != \textsc{CampVC}.rp$} \Return \label{algo:af:verirp} \Comment{\ref{criterion:4} verifying $rp$} \EndIf

    \State $hr' = \Call{Hash}{\textsc{CampVC}.txBlock, \textsc{CampVC}.nc}$ \label{algo:af:rehash}
        
    \If{!\Call{Prefix}{$hr', rp'$}} \Return \label{algo:af:rehashcheck} \Comment{\ref{criterion:5} verifying computation}\EndIf 

    \State \textbf{send} \textsc{VoteCP} $\leftarrow$ $\langle \textsc{CampVC}.V', \sigma_{S_i} \rangle$
  \EndFor

    \Statex \textcolor{blue}{$\triangledown$ As a redeemer} 
    \State \textbf{retrieve} vcBlock, txBlock, V \label{algo:red:start}
    \State $V' = V + 1$ \Comment{increments current view}
    \State $rp, ci =$ \Call{CalcRP}{$V'$, vcBlock, txBlock, myID}
    \State \textbf{stop} replication in $V$
    \State \textbf{newThread} (\textbf{upon} discovering higher view \textbf{transition} to \texttt{Follower})    
    \While{$1$} \label{algo:ar:powstarts}
        \State $nc = \Call{Gen-Nonce}{rand}$
        \Comment{Randomly generates a string}
        \State $hr= \Call{Hash}{\texttt{txBlock}, nc}$  
        \If{\Call{Prefix}{$hr$, $rp$}} \label{algo:ar:powends} \Comment{Check if $hr$ has a prefix of $rp$ 0s}
            \State \textbf{transition} to \texttt{Candidate}  
            \State \Return $V$, $V'$, rp, nc, hr, ci, \texttt{txBlock}
        \EndIf
    \EndWhile

    \Statex \textcolor{blue}{$\triangledown$ As a candidate}
    \State \textbf{reset} timer \Comment{starts an election}
    \State \textbf{broadcast} \textsc{Camp} $\leftarrow$ $\langle \texttt{conf\_QC}, V, V', rp, nc, hr, ci, \texttt{txBlock}, \sigma_{S_i} \rangle$ \label{algo:ascandy:camp}
    
    \State \textbf{newThread} (\textbf{upon} discovering higher view \textbf{transition} to \texttt{Follower}) \label{algo:ascandy:gobackaf}
    
    \While{timer}
        \For{receiving $2f+1$ \textsc{VoteCP} messages}
            \State \textbf{convert} \textsc{VoteCP}s to 
            \texttt{vc\_QC} and
            \textbf{transition} to \texttt{Leader} 
        \EndFor
    \EndWhile
    \State \textbf{transition} to \texttt{Redeemer}

    \Statex \textcolor{blue}{$\triangledown$ As a leader}
    \State \textbf{new} vcBlock\{$V'$, myID, conf\_QC, vc\_QC, vcBlock.rp, vcBlock.cp\} \label{algo:al:1}
    \State vcBlock.rp[myID], vcBlock.ci[myID] = rp, ci \label{algo:al:2}
    \State \textbf{broadcast} \texttt{vcBlock} \label{algo:al:sendnewb}
    \State \textbf{newThread} (\textbf{upon} discovering higher view \textbf{transition} to \texttt{Follower})

    \For{receiving 2f+1 \textsc{vcYes} messages}
    \State \textbf{store} vcBlock and \textbf{start} replication in $V'$ as \texttt{Leader}
    \EndFor
\end{algorithmic}
\end{algorithm}

\subsubsection{As a Follower.} \label{sec:asfollower}

Each server is initially a follower and has a timer with a random timeout. The timeout should be sufficiently greater than network latency ($\Delta$), allowing ample time for a correct leader to complete consensus (e.g., timeout range = [300, 600~ms] for a $\Delta{=}30$~ms). 

View changes can be invoked by policy-defined criteria and failure detection. The former can be implemented differently according to application specifications. For example, a throughput-threshold policy that changes a view if a leader fails to operate at an expected throughput (e.g., Aardvark~\cite{clement2009making}) or a timing policy that changes a view every $5$ minutes. The detection of a leader's failure involves both clients and servers. If a client (\texttt{c}) cannot confirm its proposed transaction (\texttt{tx}) in time through the replication protocol, it broadcasts a complaint (\textsc{Compt}) message with the proposal message (\textsc{Prop} in~\S\ref{sec:algo:replication}) it sent to the leader, including \texttt{tx}, \texttt{c}, and the client's signature $\sigma_c$, suspecting a leader's failure.

Assume a server $S_i$ operates as a follower; after receiving a \textsc{Compt} message, $S_i$ verifies and relays it to the leader and then waits for consensus to be completed (Line~\ref{algo:af:relay}). If \texttt{tx} is committed before the timer expires (Line~\ref{fol:compt-committed}), it shows that the leader is still correct. Otherwise, $S_i$ suspects that the leader or client may be faulty and starts an inspection by broadcasting a \textsc{ConfVC} message (Line~\ref{algo:af:inspection}), where $V$ is the current view number and $\sigma_{S_i}$ is the signature that $S_i$'s signs this message. The other followers, after receiving a \textsc{ConfVC} from $S_i$, check if they have received the same \textsc{Cmpt} from client $c$ (Line~\ref{algo:af:confvc}). If so, they reply with a \textsc{ReVC} message.

If $S_i$ receives $f{+}1$ replies in time (including itself), it converts them to a threshold signature with $t{=}f{+}1$, which forms a quorum certificate (\texttt{conf\_QC}). $S_i$ considers the leader faulty and starts a view change by transitioning to a redeemer. On the other hand, if $f{+}1$ replies cannot be collected in time, $S_i$ will tag client \texttt{c} as faulty. 

This failure detection mechanism prevents faulty clients, faulty servers, and their collusion from inflicting unnecessary view changes on correct followers. Since a correct client is required to broadcast its complaint to all servers, at least $2f{+}1$ correct servers can relay the complaint to the leader. Thus, a view change is confirmed by at least a correct server (Line~\ref{algo:af:revc}). Note that we do not assume DDoS attacks (in~\S\ref{sec:algo:system-model}), such as faulty clients overwhelming servers by pouring complaints. This can be handled by rate control or blacklisting tagged clients.

\subsubsection{As a Redeemer.} \label{sec:asredeemer}
After becoming a redeemer, $S_i$ retrieves the view number ($V$), the \texttt{vcBlock} of the current view, and the latest committed \texttt{txBlock} (\texttt{txBlock}.n is the highest among all \texttt{txBlock}s). $S_i$ first increments its view to $V'$ and calls the reputation engine to get its $rp$ and $ci$ for view $V'$. Next, it computes a hash puzzle (similar to Proof-of-Work~\cite{nakamoto2019bitcoin}): it generates a random string ($nc$) and hashes the combination of $nc$ and \texttt{txBlock} until the hash result ($hr$) has a prefix of $rp$ zero bytes (e.g., $hr$=``0000966sv0d3...'' under $rp{=}4$). Thus, the higher the $rp$ is, the more iterations it takes to find a prefix with $rp$ leading zeros.

Servers with a higher $rp$ will spend more time and energy to ``redeem themselves'' from the imposed work, while servers with a lower $rp$ can complete the computation and transition to a candidate more quickly.

In addition, a redeemer transitions back to follower when it discovers a leader of a higher view. If the leader's \texttt{vcBlock} is valid (see \S\ref{sec:asleader}), it indicates that the redeemer is out-of-sync. The redeemer will abort ongoing computation and operate as a follower in the higher view.

\subsubsection{As a Candidate.} \label{sec:ascandidate}

After becoming a candidate, $S_i$ broadcasts a campaign message (Line~\ref{algo:ascandy:camp}), where \texttt{conf\_QC} was collected when confirming this view change (Line~\ref{algo:af:confqc}), $V$ to $\texttt{txBlock}$ are the results after the redeemer state, and $\sigma_{S_i}$ is the signature that $S_i$ signs this message with. Then, $S_i$ waits for votes from the other servers. The other servers, operating as followers, vote for $S_i$ (Line~\ref{algo:fol:pro3}) with the following criteria: 

\begin{enumerate} [label=\textbf{C\arabic*}]
    \item \label{criterion:1}
    The follower has not voted in this view (\textsc{Camp}.$V'$).
    
    \item \label{criterion:2}
    The threshold of \textsc{Camp}.\texttt{conf\_QC} is $f{+}1$.
    
    \item \label{criterion:3}
    The candidate's replication is at least as up-to-date as the follower's.
    
    \item \label{criterion:4}
    The candidate's $rp$ can be recalculated and verified.
    
    \item \label{criterion:5}
    The candidate's performed computational work aligns with its $rp$; i.e., \textsc{Camp}.$hr$ has $rp$ leading zeros.
\end{enumerate}

\ref{criterion:1} enforces that a server votes at most once in a view, and thus guarantees Property~\ref{vc:p1} that at most one leader can be elected in a given view. \ref{criterion:2} guarantees that the current view change is necessary and was confirmed by at least one correct server.

Since BFT consensus operates in $QC$s of size $2f+1$, up to $f$ servers can be correct but stale (fallen behind in their logs including \texttt{txBlock}s and \texttt{vcBlock}s). For example, if $S_1$ fails in Figure~\ref{fig:passive-rotation} and $S_3$ becomes a candidate, $S_4$ and $S_3$ have identical logs. However, while $S_2$ is correct, it has stale logs at the time of $S_1$'s crash. Therefore, stale servers must sync to update their logs before verifying requests from candidates. To achieve this, the \textsc{SyncUp} function acquires needed blocks from the candidate:
{\small
\begin{algorithmic}
  \Function{SyncUp}{btype, end} \Comment{btype is a block interface}
    \State $start$ = btype.id \Comment{id = view/n in vcBlocks/txBlocks}
    \State \textbf{acquire} blocks[] from $start$ to $end$ from remote
    \State \textbf{validate} all blocks in blocks[] through their $QC$s
    \State btype = blocks[:-1] \Comment{set myBlock to the latest block}
  \EndFunction
\end{algorithmic}
}
If followers fall behind in view changes, they call \textsc{SyncUp} to acquire missing \texttt{vcBlock}s (Line~\ref{algo:af:syncvc}). Then, they check if the candidate's replication is at least up-to-date as themselves (Line~\ref{algo:candy:replistatus} for~\ref{criterion:3}), and sync up replication if falling behind (Line~\ref{algo:af:synctx}). \ref{criterion:3} enforces the election of a candidate that has the most up-to-date log, which ensures Property~\ref{vc:property:or}.

After syncing up, followers can verify the candidate's $rp$ and associated computation. They use the same calculation scheme by calling Algo.~\textsc{CalcRP} with the candidate's ID, where $rp$ and $ci$ should be reproduced (Line~\ref{algo:af:rcalrp} to~\ref{algo:af:verirp}). If so, \ref{criterion:4} is satisfied. Then, followers verify the candidate's computational work. They reproduce the hash result and check if the result has a prefix of $rp$ leading zero bytes (Line~\ref{algo:af:rehash} to~\ref{algo:af:rehashcheck}). If so, \ref{criterion:5} is satisfied, thereby ensuring Property~\ref{vc:property:verify}. Note that followers only hash once ($\mathcal{O}(1)$) to verify the computation. Finally, followers send a vote back to the candidate. Proofs of~\ref{vc:p1},~\ref{vc:property:or}, and~\ref{vc:property:verify} are provided  in Appendix~\ref{sec:ap:correctness}.

The candidate becomes the new leader if it can collect $2f{+}1$ \textsc{VoteCP}s in time. It then coverts the votes to a \texttt{vc\_QC} with a threshold of $2f{+}1$ and declares leadership. During this process, the candidate may find itself out-of-sync if it discovers a leader operating in a higher view; it will abort the election and transition back to a follower (Line~\ref{algo:ascandy:gobackaf}).

On the other hand, if the candidate neither becomes a leader nor transitions back to a follower when its timer expires, split votes may have occurred, where multiple candidates campaigning for the same view and collecting partial votes (because of \ref{criterion:1}), similar to Raft's split votes~\cite{ongaro2014search}. In this case, the candidate transitions back to a redeemer with incremented view $V'{+}1$ (Line~\ref{algo:red:start}) and starts a new campaign. It is worth noting that this situation is extremely rare, especially with randomized timers (in~\S\ref{sec:asfollower}). Our evaluation in~\S\ref{sec:evaluation} shows that split votes never occurred in 10,000 view changes with just $50$ms of randomization.

\subsubsection{As a Leader} \label{sec:asleader}

After becoming a leader, $S_i$ prepares a new \texttt{VcBlock} with the parameters shown in Figure~\ref{fig:blocks}. It inherits the old view's \texttt{vcBlock} (view $V$) with its updated $rp$ and $ci$ (Line~\ref{algo:al:2}) and broadcasts the new \texttt{vcBlock}. \textbf{Note that only the elected leader may have a change in its $rp$ in view changes; unsuccessful view change attempts will not result in an $rp$ change.} 

{\footnotesize
\algrenewcommand\algorithmicfor{\textbf{upon}}
\begin{algorithmic}
\Procedure{Receiving}{\texttt{newVcBlock}} \label{algo:af:pro4} \Comment{\texttt{vcBlock} sent by the leader}
    \State \textbf{validate} \texttt{newVcBlock} and \textbf{compare} with \texttt{myVcBlock} \label{algo:af:rvcb:compare} 

    \State \textbf{send} \textsc{vcYes} $\leftarrow$ $\langle \texttt{newVcBlock.v}, \sigma_{S_i} \rangle$
    
    \State \textbf{stop} \texttt{myVcBlock}.v and \textbf{start} 
    operating in \texttt{newVcBlock}.v as follower  
\EndProcedure
\end{algorithmic}
}

The other non-leader servers (i.e., followers, redeemers, and candidates), follow the above steps to verify \texttt{newVcBlock}. They validate the $QC$s and compare the reputation segment of \texttt{newVcBlock} with that of \texttt{myVcBlock} (i.e., the current \texttt{vcBlock} of $V$). If the only change is the leader's $rp$ and $ci$, servers adopt \texttt{newVcBlock} and send a \textsc{vcYes} message to the leader. When the leader collects $2f{+}1$ \textsc{vcYes}s, the consensus for \wct<1> \emph{deciding a new leader} for the new view and \wct<2> \emph{updating the new leader's reputation penalty} has been achieved; then, normal operation resumes under the new leadership.

\textbf{A note on using Proof-of-Work (PoW).} 
In \algo, \textbf{PoW is never involved in replication.} It is only used as an implementation of the reputation penalty to reduce the probability of electing suspected faulty servers in view changes. Alternatively, Verifiable Delay Functions (VDF)~\cite{boneh2018verifiable} can also be used to delay high penalty servers for participating view changes. Our use of PoW imposes a computational cost on attackers without overburdening correct servers (``let bad guys pay''). Our evaluation shows that the cost for correct servers is negligible (less than 20~ms for $rp{<}5$) but becomes significantly higher for attackers (hours for $rp{>}8$). While VDF is more environmentally friendly, PoW is an efficient and economic deterrent to malicious attacks.

\subsubsection{Refresh penalties}
\label{sec:optimizations}
In the partial synchrony model (\S~\ref{sec:algo:system-model}), when GST is sufficiently long, it may trigger timeouts on non-faulty servers. This may cause non-faulty leaders to get penalized in the long run.
\algo allows a refresh on imposed $rp$ when at least $f{+}1$ non-faulty servers get penalized with their $rp$ exceeding a \emph{threshold} ($\pi$). The refresh process for a server $S_i$ is as follows.

\begin{enumerate}
    \item $S_i$ broadcasts a \textsc{Ref} message: $\langle \textsc{Ref}, V, \sigma_{S_i} \rangle$.
    
    \item Upon receiving $2f+1$ \textsc{Ref}s from different servers (including itself), $S_i$ converts them to a \texttt{rs\_QC} and set its $rp$ and $ci$ to the initial values. It then broadcasts a \textsc{Rdone} message: $\langle \textsc{Rdone}, \texttt{rs\_QC}, V, rp, ci, \sigma_{S_i} \rangle$.
    
    \item After receiving a \textsc{Rdone} message from $S_i$, the other servers verify \texttt{rs\_QC} and update $S_i$'s $rp$ and $ci$ in the current \texttt{VcBlock}. 
\end{enumerate}

This refresh mechanism ensures that when a server refreshes its penalty, there have been at least $f+1$ correct servers (in \texttt{rs\_QC}) whose $rp$ has exceeded $\pi$. The refresh sets both $rp$ and $ci$ to their initial values, relieving the imposed potential computational work as well as refreshing the compensation of $\delta_{tx}$ and $\delta_{vc}$ for future $rp$ calculations.

\subsection{The replication protocol}
\label{sec:algo:replication}

With Properties~\ref{vc:p1}~\ref{vc:property:or} and~\ref{vc:property:verify} in the view-change protocol, \algo can achieve consensus with optimistic responsiveness~\cite{pass2018thunderella} using a standard two-phase replication protocol, which are the ordering phase and commit phase. In replication, \textit{servers never respond to a leader that has a lower view.} A replication consensus instance works as follows.

\noindent \textbf{Invoking a consensus service (1 round):}
A client broadcasts a proposal $\langle \textsc{Prop}, t, d, c, \sigma_c, tx\rangle$ to all servers, including a unique timestamp ($t$), a transaction ($tx$), the digest of the transaction ($d$), its ID ($c$), and its signature ($\sigma_c$) that signs $t$, $d$, and $c$. It then waits for this proposal to be committed. 

\noindent \textbf{Phase 1: constructing \texttt{ordering\_QC} (2 rounds):}

\begin{itemize}
\item The leader starts a consensus instance for $tx$ when \wct<1> it receives a \textsc{Prop} message from a client, or \wct<2> $f+1$ \textsc{Compt} messages from different servers. The leader then assigns a unique sequence number $n$ to \textsc{Prop} and broadcasts an ordering message: $\langle \textsc{Ord}, \langle \textsc{Prop} \rangle, n, V, \sigma_{S_i} \rangle$.

\item Followers verify the received \textsc{Ord} message by checking that $n$ has not been used. Then, they send a reply to the leader with their signatures.

\item The leader waits for $2f+1$ replies and converts them (of size $\mathcal{O}(n)$ in total) to a threshold signature (of size $\mathcal{O}(1)$), which forms the $QC$ of \texttt{ordering\_QC}.
\end{itemize}

\noindent \textbf{Phase 2: constructing \texttt{commit\_QC} (3 rounds)}

\begin{itemize}
\item The leader then broadcasts a \textsc{Cmt} message with the obtained $QC$: $\langle \textsc{Cmt}, \texttt{ordering\_QC}, V, \sigma_{S_i} \rangle$

\item Followers verify \texttt{ordering\_QC}'s threshold and then send replies to the leader with their signatures.

\item The leader waits for $2f{+}1$ replies to form \texttt{commit\_QC} and prepares a \texttt{txBlock} (shown in Figure~\ref{fig:blocks}) by setting the block agreement and transaction fragments accordingly. Then, it broadcasts  \texttt{txBlock} and sends a \textsc{Notif} message to the client.
\end{itemize}

\subsubsection*{Terminating consensus instance (1 round)}

\begin{itemize}
    \item Followers verify the received \texttt{txBlock} and then send a \textsc{Notif} message to the client.
    
    \item If the client can receive $f+1$ \textsc{Notif}s before its timer expires, it considers $tx$ committed. Otherwise, it complains to the servers (in \S\ref{sec:algo:active-vc}).
\end{itemize}

The replication protocol has a message complexity of $O(n)$ and a time complexity of $7$ (rounds). \algo achieves optimistic responsiveness (OR) using a two-phase protocol, which is more efficient than HotStuff~\cite{yin2019hotstuff}'s three-phase protocol. HotStuff's additional phase is to sync up non-faulty servers about the commit result, as its passive VC protocol blindly rotates leadership. In contrast, \algo's active VC protocol allows servers to elect the most up-to-date candidate, leading to high performance due to the reduced messages and rounds needed (see~\S\ref{sec:evaluation}).

\section{Correctness argument}
\label{sec:correctnes}

This section sketches the correctness arguments including safety and liveness. We denote the assumed computation bound of faulty servers in our system model as $\gamma$ (\S\ref{sec:algo:system-model}).

\begin{theorem}[Validity]
In each consensus instance, if all servers have received the same $tx$, then any $tx$ committed by a non-faulty server must be that common $tx$.
\end{theorem}

\begin{proof}
Each committed $tx$ must have been endorsed by a \texttt{commit\_QC}. A server that signs in \texttt{commit\_QC} must have verified a corresponding \texttt{ordering\_QC}. Since \texttt{ordering\_QC} is signed by $2f{+}1$ servers, a committed $tx$ must be the common value that has been seen by at least $2f{+}1$ servers.
\end{proof}

\begin{lemma}
\label{lemma:f+1uptodate}
At least $f+1$ non-faulty servers are up-to-date.
\end{lemma}

\begin{proof}
Lemma~\ref{lemma:f+1uptodate} is straightforward. Both \texttt{vc\_QC} in \texttt{vcBlock}s and \texttt{commit\_QC} in \texttt{txBlock}s have a threshold of $2f+1$, which can include up to $f$ faulty servers. Therefore, there are at least $f+1$ non-faulty servers included in these $QC$s and thus are up-to-date in replication in the highest view.
\end{proof}

\begin{lemma}
\label{lemma:oneleader}
In any given view change, at least $f+1$ non-faulty servers are eligible for being elected as the leader. 
\end{lemma}

\begin{proof}
We denote the set of non-faulty and up-to-date servers in Lemma~\ref{lemma:f+1uptodate} as $\mathcal{S}_1$, and non-faulty but stale servers as $\mathcal{S}_2$ ($|\mathcal{S}_1|+|\mathcal{S}_2|=2f+1$). In the worst case, $|\mathcal{S}_1|=f+1$ and $|\mathcal{S}_2|=f$. Since $\mathcal{S}_2$ can always sync up to a more up-to-date candidate (in \S\ref{sec:ascandidate}), $\forall S_i \in \mathcal{S}_2$ can vote for $\forall S_j \in \mathcal{S}_1$. Thus, $\forall S_j \in \mathcal{S}_1$ are eligible for an election.
\end{proof}

The major difference between the passive and active VC protocols is that faulty servers can actively campaign for leadership and replace a correct leader. In the passive protocol, $f$ faulty servers cannot usurp leadership from a correct leader when they are not the scheduled leader. In the active protocol, all servers can campaign for leadership, which gives faulty servers the opportunity to replace a correct leader. The orchestration of \algo's voting-based leader election and reputation mechanisms makes great effort to suppress Byzantine servers from being elected and reduces the possibility of Byzantine leaders over time. 

\begin{lemma}
\label{lemma:repossess}
Faulty servers cannot repossess leadership indefinitely without making progress in replication.
\end{lemma}

\begin{proof}
If a faulty server does not make progress in replication, it cannot get compensated as its $\delta_{tx}$ remains $0$. Thus, its $rp$ keeps increasing. In addition, if a faulty server makes limited replication and stops, after it gets compensated, its $ci = ti$. If it no longer make replication progress, its $\delta_{tx}=0$ and its $rp$ keeps increasing. When the required computational work exceeds the faulty server's computation capability $\gamma$, the faulty server cannot transition to a candidate and thus will never be elected. 
\end{proof}

\begin{theorem}[Liveness] \label{theorem:liveness:paper}
After GST, a non-faulty server eventually commits a proposed client request.
\end{theorem}

\begin{proof}[Proof (sketch)]
In any given time, leadership is in one of the following two conditions: \wct<1> leadership is controlled by $f$ faulty servers, or \wct<2> leadership is released by $f$ faulty servers.

In \wct<1>, with Lemma~\ref{lemma:repossess}, faulty leaders must at some point start to conduct replication. Otherwise, they cannot  control the leadership indefinitely. When they start to conduct replication, they become temporary non-faulty leaders.

In \wct<2>, with Lemma~\ref{lemma:oneleader}, a leader will eventually be elected from up-to-date and non-faulty servers. Thus, after GST, a client request will eventually be committed by all non-faulty servers. Therefore, in both cases, \algo ensures that a client eventually receives replies to its request after GST.
\end{proof}

\begin{theorem}[Safety]
\label{theorem:safety}
Non-faulty servers do not decide on conflicting blocks. That is, non-faulty servers do not commit two \texttt{txBlock}s at the same sequence number $n$.
\end{theorem}

\begin{proof}[Proof (sketch)]
With Property~\ref{vc:p1}, no view has more than one leader. Next, we prove this theorem by contradiction. We use the partition of servers in Lemma~\ref{lemma:oneleader} and denote faulty servers as $\mathcal{S}_f$. We claim there are \texttt{txBlock} and $\texttt{txBlock}_{\diamond}$, both committed with sequence number $n$. 

In this case, \texttt{commit\_QC} and $\texttt{commit\_QC}_{\diamond}$ are both signed by $2f+1$ servers. Say \texttt{commit\_QC} is signed by servers in $\mathcal{S}_1 \cup \mathcal{S}_f$. Then, servers in $\mathcal{S}_1$ cannot sign $\texttt{commit\_QC}_{\diamond}$ with $n$. Although faulty servers in $\mathcal{S}_f$ can double commit, $\texttt{commit\_QC}_{\diamond}$ can only find servers in $\mathcal{S}_f \cup \mathcal{S}_2$ ($|\mathcal{S}_f|{+}|\mathcal{S}_2|{=}2f$) to sign it, which is not sufficient to form a $QC$ of size $2f+1$. Therefore, $\texttt{commit\_QC}_{\diamond}$ cannot be formed, which contradicts our claim.
\end{proof}

Due to space limitations, here we provided only proof sketches and refer readers to the Appendix, where we provide the complete proofs with visualized analysis (Appendix~\S\ref{sec:ap:correctness}), collected Q\&A from researchers, CS/ECE students, and developers (Appendix~\S\ref{sec:ap:collectedq}), and various examples (Appendix~\S\ref{sec:ap:examples}) to aid understanding.

\section{Evaluation}
\label{sec:evaluation}

We implemented \algo in Golang and deployed it on $4$, $16$, $31$, $61$, and $100$ VM instances on a popular cloud platform. Each instance includes a machine with $4$ vCPUs supported by $2.40$ GHz Intel Core processors (Skylake) with a cache size of $16$ MB, $15$ GB of RAM, and $75$ GB of disk space running on Ubuntu $18.04.1$ LTS. The TCP/IP bandwidth measured by \texttt{iperf} is around $400$ MB/s with a raw network latency between two instances less than $2$~ms. We use the following notations to report on the results.

\begin{tabular}{ll}
     $n$ & The number of VMs (scales) \\
     $\beta$ & The number of transactions in a batch (batch size) \\
     $d$ & Emulated network delays (ms) using \texttt{Netem} \\
     $m$ & The message size (e.g., $m=32$ bytes) \\
\end{tabular}

\begin{figure*}[t]
\minipage{0.25\textwidth}
    \centering
    \includegraphics[width=\linewidth]{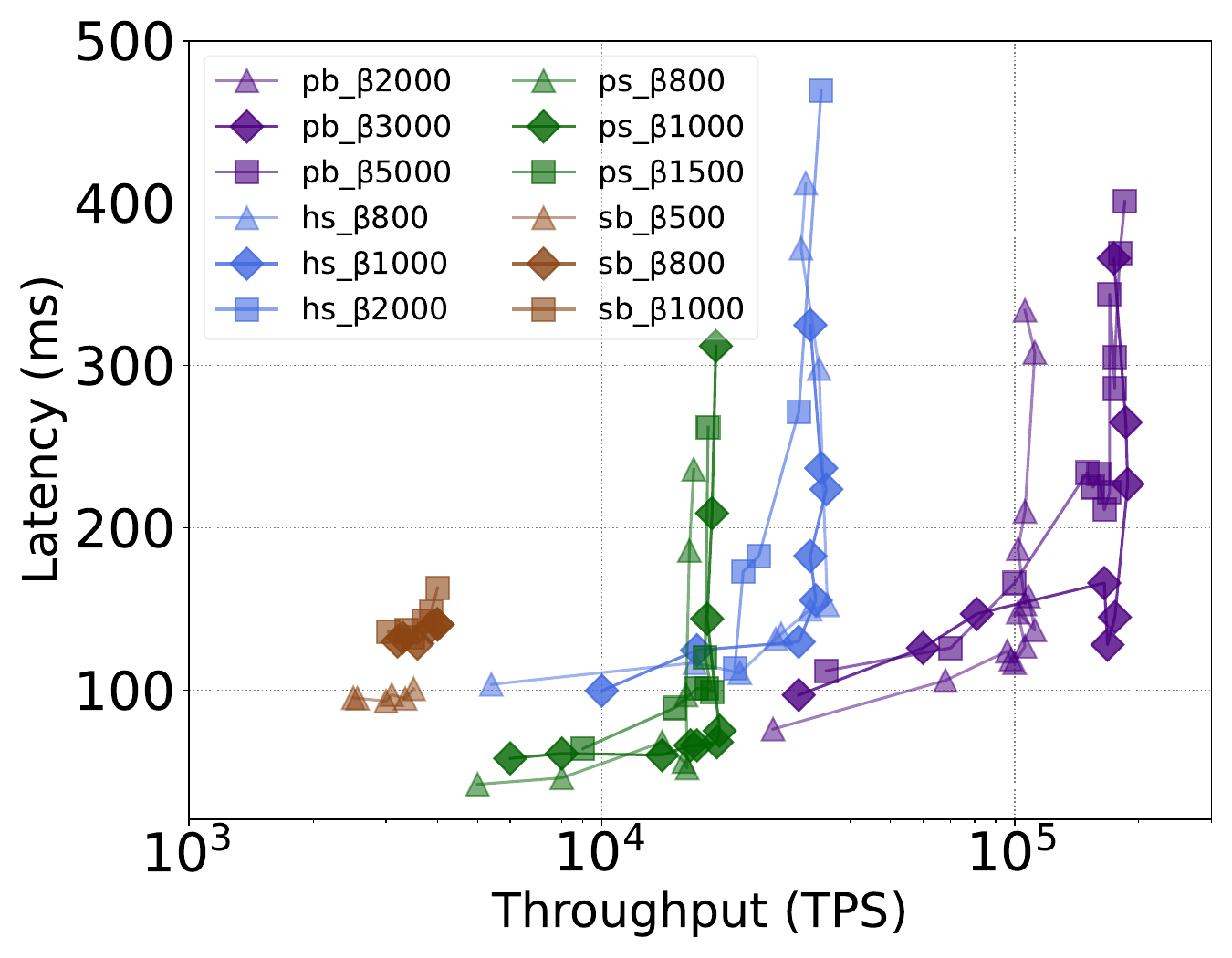}
    \caption{Performance under batching ($n{=}4$ and $m{=}32$).}
    \label{fig:eva:normalpeak}
\endminipage
\hfill
\minipage{0.48\textwidth}
    \centering
    \includegraphics[width=.98\linewidth]{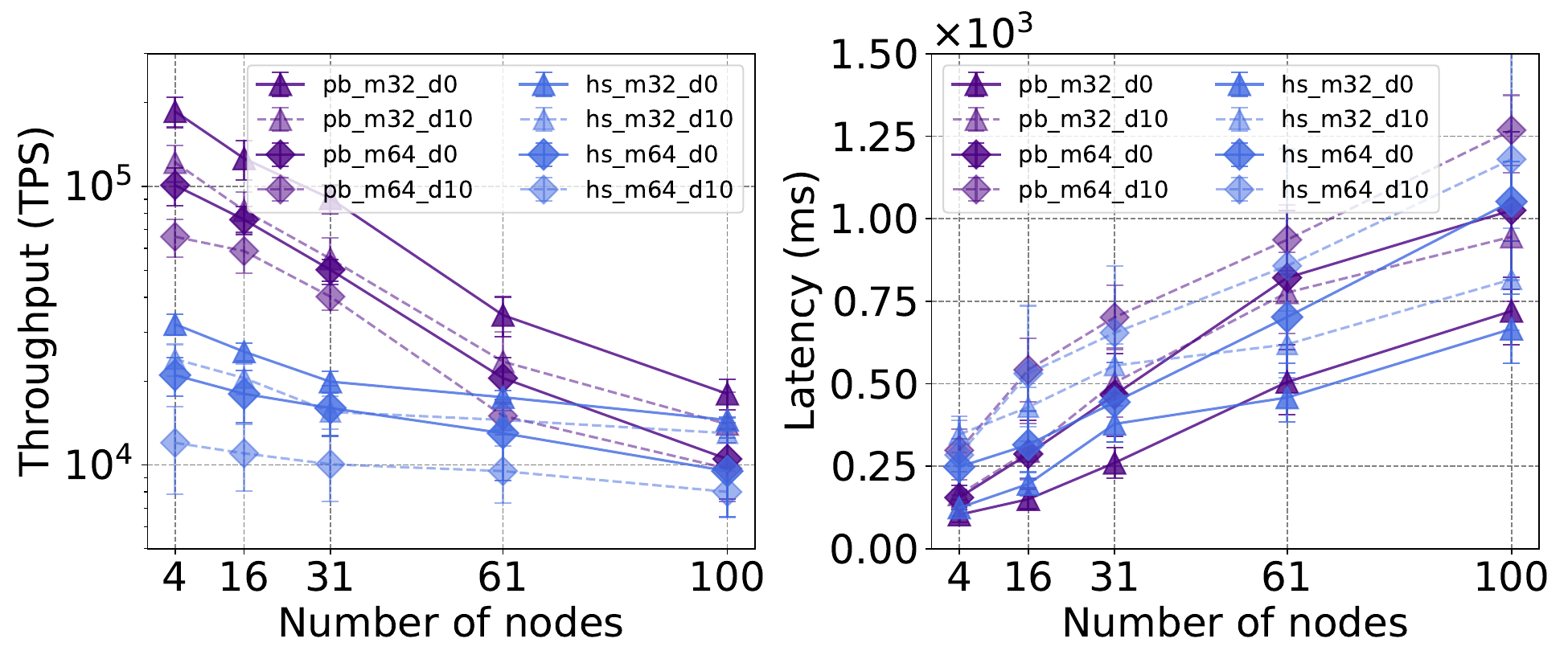}
    \caption{Throughput (left) and latency (right) under increasing system scales ($m=32$ and $64$, $d=0$ and $10\pm5$~ms).}
    \label{fig:eva:normalscalability}
\endminipage
\hfill
\minipage{0.25\textwidth}
    \centering
    \includegraphics[width=\linewidth]{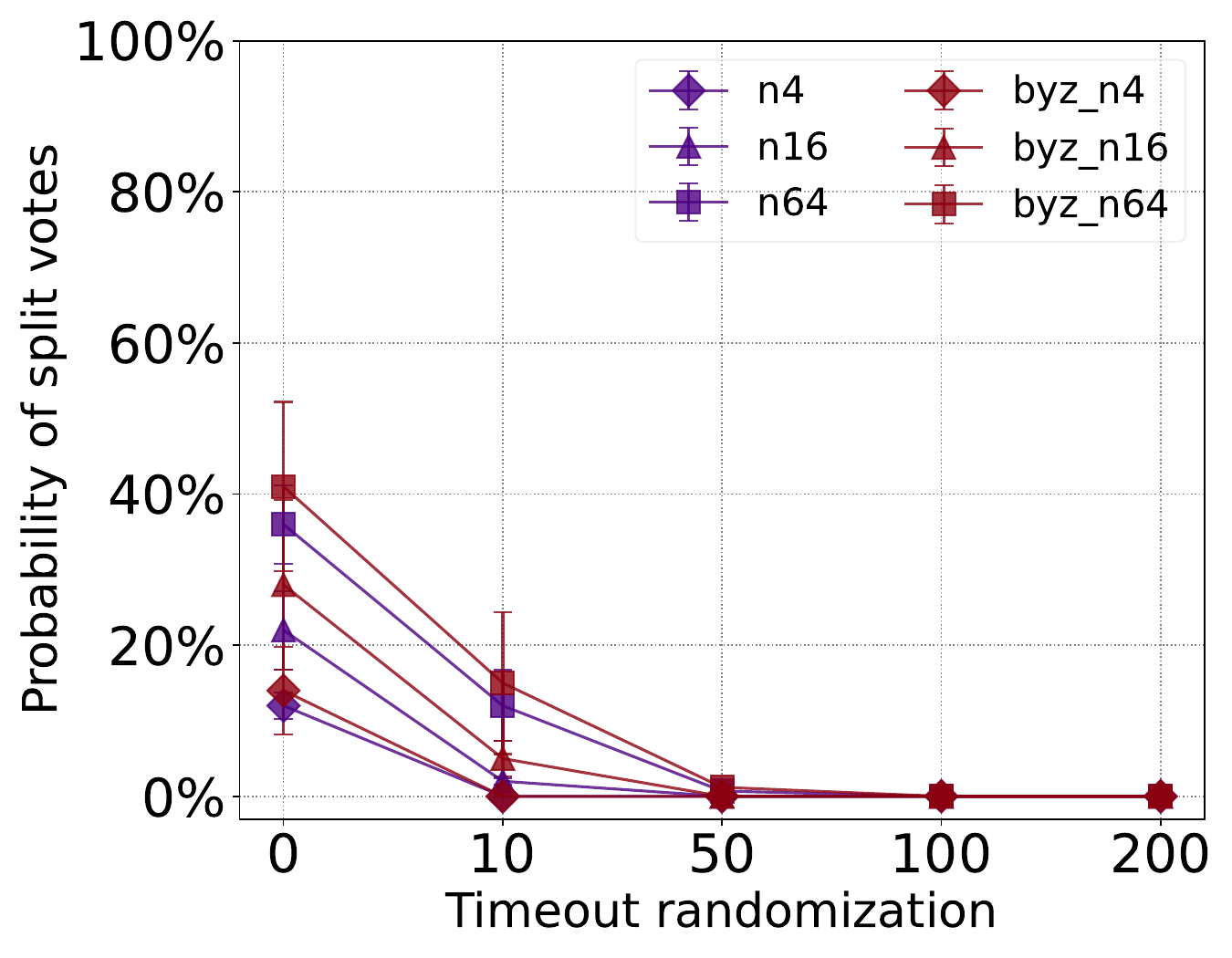}
    \caption{Split votes under different timeout randomization.}
    \label{fig:eva:splitvotes}
\endminipage
\end{figure*}

\subsection{Performance under normal operation}
\label{sec:eva:normal-op}

Performance is reported in terms of throughput and latency. Since linear BFT algorithms have shown significant performance advantage compared to non-linear algorithms (e.g., HotStuff~\cite{yin2019hotstuff} ($\mathcal{O}(n)$) $>$ PBFT~\cite{castro1999practical}($\mathcal{O}(n^2)$) $>$ RBFT~\cite{aublin2013rbft}($\mathcal{O}(n^3)$)), we conducted end-to-end comparisons for \algo (\texttt{pb}) against three linear BFT algorithms: SBFT~\cite{concord} (\texttt{sb}), HotStuff~\cite{libhotstuff} (\texttt{hs}), and Prosecutor~\cite{zhang2021prosecutor} (\texttt{pr}) using their original implementations. Throughput was measured in transactions per second (TPS) on servers, i.e., the number of requests committed in one second. Latency was measured on clients from when a request is sent to when $f{+}1$ \textsc{notif}s are received. 

\textbf{Peak performance.} 
The peak performance was measured when $n{=}4$. Clients generated random requests of $m{=}32$ bytes and waited for one request to complete before sending the next one. We kept increasing batch sizes for each algorithm to find the batch size that resulted in the highest throughput while maintaining a low latency. Under each batch size, we kept deploying more clients until their generated workloads were sufficient (until an elbow of a curve occurs (Figure~\ref{fig:eva:normalpeak})).

\algo outperforms its baselines with a peak performance at a throughput of 186,012 TPS and a latency of 166~ms at $\beta {=} 3000$. Its peak performance is $5.2\times$ higher than HotStuff's, which peaked at $35,428$ TPS in $129$~ms at $\beta {=} 1000$. The high performance is attributed to the reduced phases in replication while achieving optimistic responsiveness. Prosecutor performed at a similar throughput compared to HotStuff with a lower latency, and SBFT peaks at $4,872$ TPS in $148$~ms at $\beta {=} 800$ (similar results are reproduced by~\cite{zhang2020byzantine}).

\textbf{Scalability.} We evaluated the performance of the two best-performing algorithms (i.e., \texttt{pb} and \texttt{hs}) at increasing system scales with two workloads ($m{=}32$ and $64$ bytes), choosing their best batch sizes when $n{=}4$ (i.e., $\beta {=}1000$ for \texttt{hs} and $\beta{=}3000$ for \texttt{pb}). In addition to the raw network latency ($d{=}0$), we implemented additional network delays of $d{=}10 {\pm} 5$~ms at normal distribution using \texttt{netem} to emulate a higher network latency. The results show that the throughput of both algorithms decreases while their latencies increase with cluster sizes (shown in Figure~\ref{fig:eva:normalscalability}). Under the emulated network delay, consensus latency increased significantly as delayed messaging prolongs packing requests into batches and also results in a high variance.

\subsection{Performance under failures}
\label{sec:eva:under-failures}
We also conducted experiments to evaluate \algo's performance under failures. Since we cannot simulate all types of Byzantine failures, we have considered the following four common attacks.

\begin{enumerate} [label=\textbf{F\arabic*}]
    \item (Timeout attacks) Faulty servers set their timeouts to the same of $f$ randomly picked correct servers. \label{byz:timeout}  
    
    \item (Quiet participants) Faulty servers do not respond to any request (similar to crash/send omission failures). \label{byz:quiet}
    
    \item (Equivocation) Faulty servers reply to a quest by sending back erroneous messages.
    \label{byz:equivocation}
    
    \item (Repeated view-change attacks) Faulty servers campaign for leadership when they are not the leader. \label{byz:vcattacks}
\end{enumerate}

\textbf{Split votes under~\ref{byz:timeout}.}
The nature of active view changes allows servers to campaign for leadership. In theory, multiple servers can become candidates simultaneously and may cause split votes (discussed in~\S\ref{sec:ascandidate}) prolonging an election. Nevertheless, for this scenario to occur, competing candidates must detect a leader's failure $+$ finish their computations $+$ arrive their campaign requests at other servers in the same period, which is extremely rare when correct servers randomize their timeouts.

We accessed $10,000$ view changes under $n{=}4$, $16$, and $64$ with increasing amount of randomization ($\epsilon$) and set timeouts from [$800$, $800{+}\epsilon$~ms]. We observe that randomization significantly reduced the occurrence of split votes (shown in Figure~\ref{fig:eva:splitvotes}). Without failures, split votes stopped occurring with just $\epsilon {=} 50$~ms while timeout attacks (\ref{byz:timeout}) only slightly increased the occurrence and could not inflict split votes when $\epsilon {>} 100$~ms.

\textbf{Performance under \ref{byz:quiet} and \ref{byz:equivocation}.} HotStuff uses a passive view-change protocol inherited from PBFT~\cite{castro1999practical}, which is also used by numerous leader-based BFT algorithms, such as Zyzzyva~\cite{kotla2007zyzzyva}, SBFT~\cite{gueta2019sbft}, Aardvark~\cite{clement2009making}, and RBFT~\cite{aublin2013rbft}. This makes HotStuff an ideal baseline for comparing the performance of passive protocols with \algo's active protocol. The trends in performance changes can provide insights into the effectiveness of the passive vs. active comparison. 

To show the performance under frequent view changes, we implemented a simple timing policy; i.e., each server starts a view change every $x$~time in a view. Thus, a server in both algorithms triggers a view change if \wct<1> a leader failure is reported, or \wct<2> the current view has operated for $x$ time. We set $x{=}10$ (more frequent, higher decentralization) and $30$ seconds (less frequent, lower decentralization), denoted by (\texttt{r10}) and (\texttt{r30}), respectively. We set HotStuff's initial timeout to $1$ second and our timeouts from [$800$, $1200$~ms]. 

We arbitrarily chose $f{=}1$ and $f=1, 3, 5$ servers when $n{=}4$ and $n{=}16$ to perform~\ref{byz:quiet} and~\ref{byz:equivocation} and measured the throughput (shown in Figure~\ref{fig:eva:quietandequiv}). We ran each algorithm for $20$~min\footnote{Under $n=16$ and \texttt{r/30S}, to rotate leadership in a full circle in the passive view-change protocol, it requires at least $30~s\times16= 8~min$.}. 
Generally, \ref{byz:equivocation} caused a higher drop in throughput than~\ref{byz:quiet}, as it sends erroneous messages consuming the bandwidth.

\begin{figure*}[t]
\minipage{0.49\textwidth}
    \centering
    \includegraphics[width=.98\linewidth]{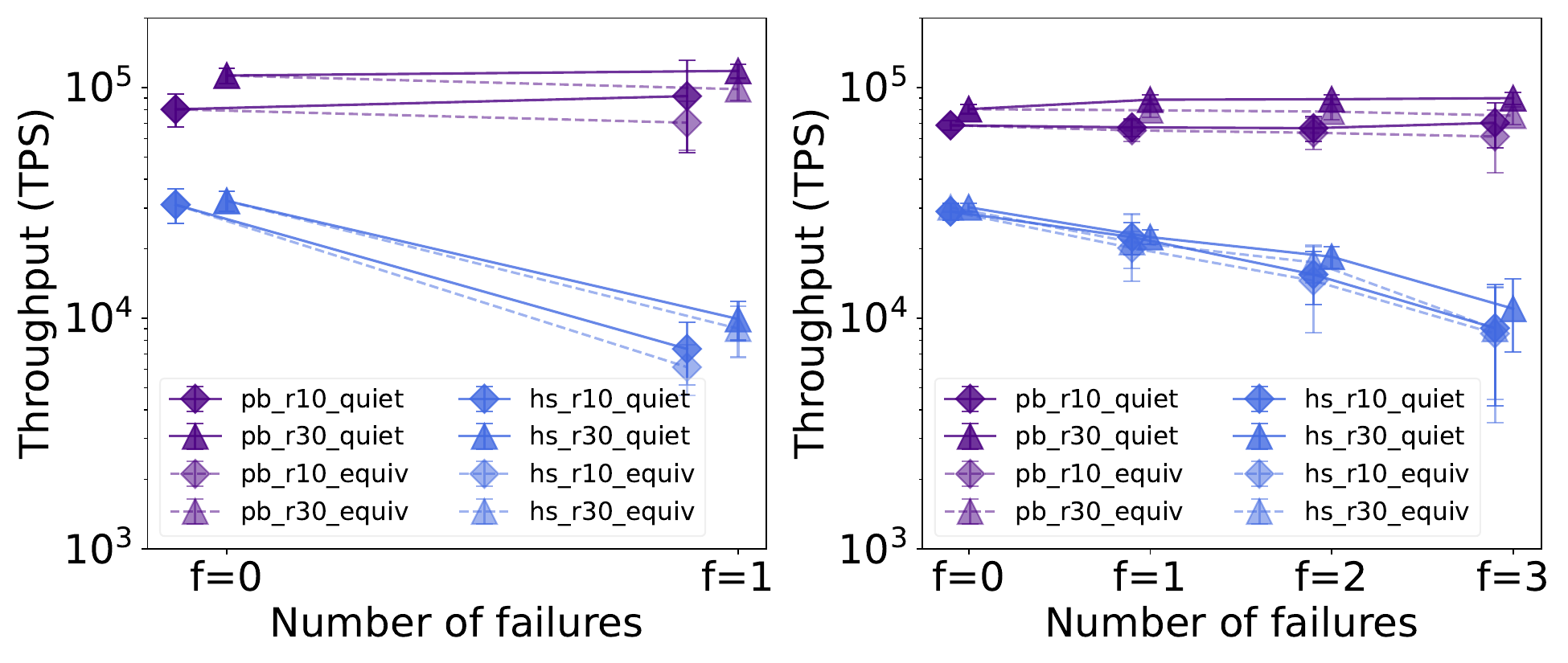}
    \caption{Throughput comparisons under quiet and equivocation attacks (\ref{byz:quiet}+\ref{byz:equivocation}) in $n{=}4$ (left) and $n{=}16$ (right).}
    \label{fig:eva:quietandequiv}
\endminipage \hfill
\minipage{0.49\textwidth}
    \centering
    \includegraphics[width=.98\linewidth]{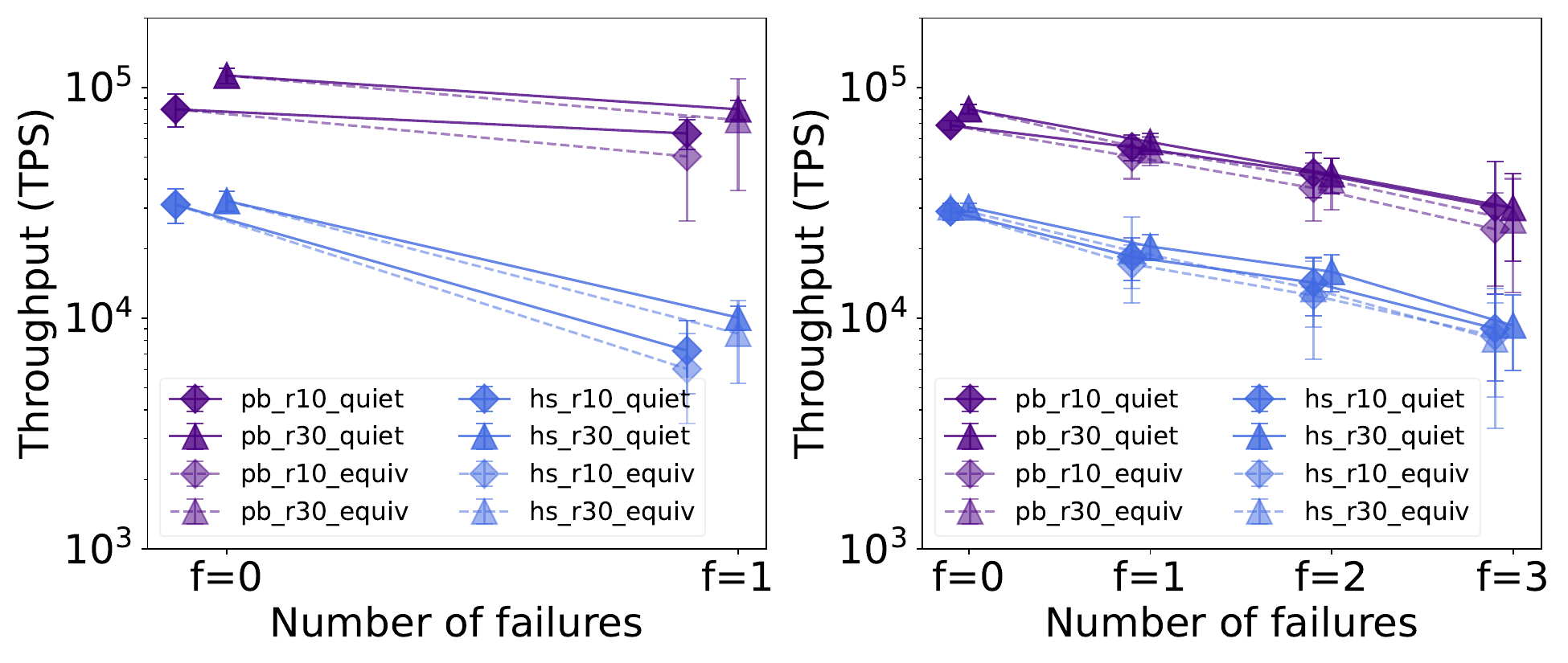} 
    \caption{Throughput comparisons under repeated VC attacks (\ref{byz:vcattacks}+\ref{byz:quiet} and \ref{byz:vcattacks}+\ref{byz:equivocation}) in $n{=}4$ (left) and $n{=}16$ (right).}
    \label{fig:eva:byzfaults}
\endminipage
\end{figure*}

HotStuff encountered $f$ faulty servers being assigned as a leader by its passive VC protocol. Each faulty leader cost the system around $1.2$~s ($1$~s timeout + $200$~ms voting and switching leaders) with no transactions replicated. When $n{=}4$, HotStuff's throughput dropped by nearly $62\%$, from $32,234$ TPS ($f{=}0$) to $12,051$ TPS ($f{=}1$) at \texttt{hs\_r30\_{quiet}}.  Additionally, more frequent rotations resulted in a higher decrease in throughput under both~\ref{byz:quiet} and~\ref{byz:equivocation}. 

\begin{figure*}[t]
\minipage{0.24\textwidth}
    \centering
    \includegraphics[width=\linewidth]{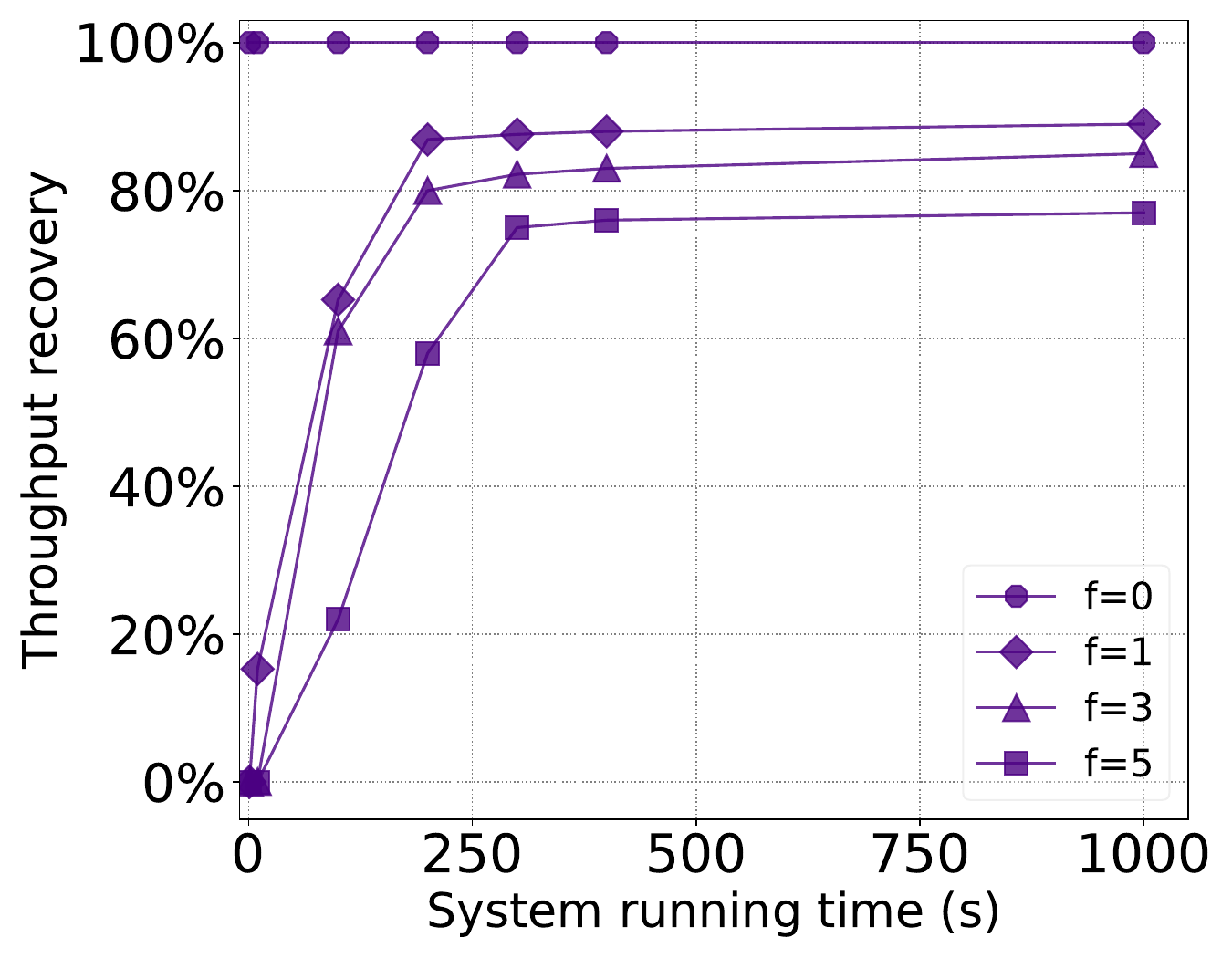}
    \caption{Improving TPS under \ref{byz:vcattacks}+\ref{byz:quiet} in \texttt{pb\_r10\_{quiet}}.}
    \label{fig:eva:byz:vbimproving}
\endminipage \hfill
\minipage{0.24\textwidth}
    \centering
    \includegraphics[width=\linewidth]{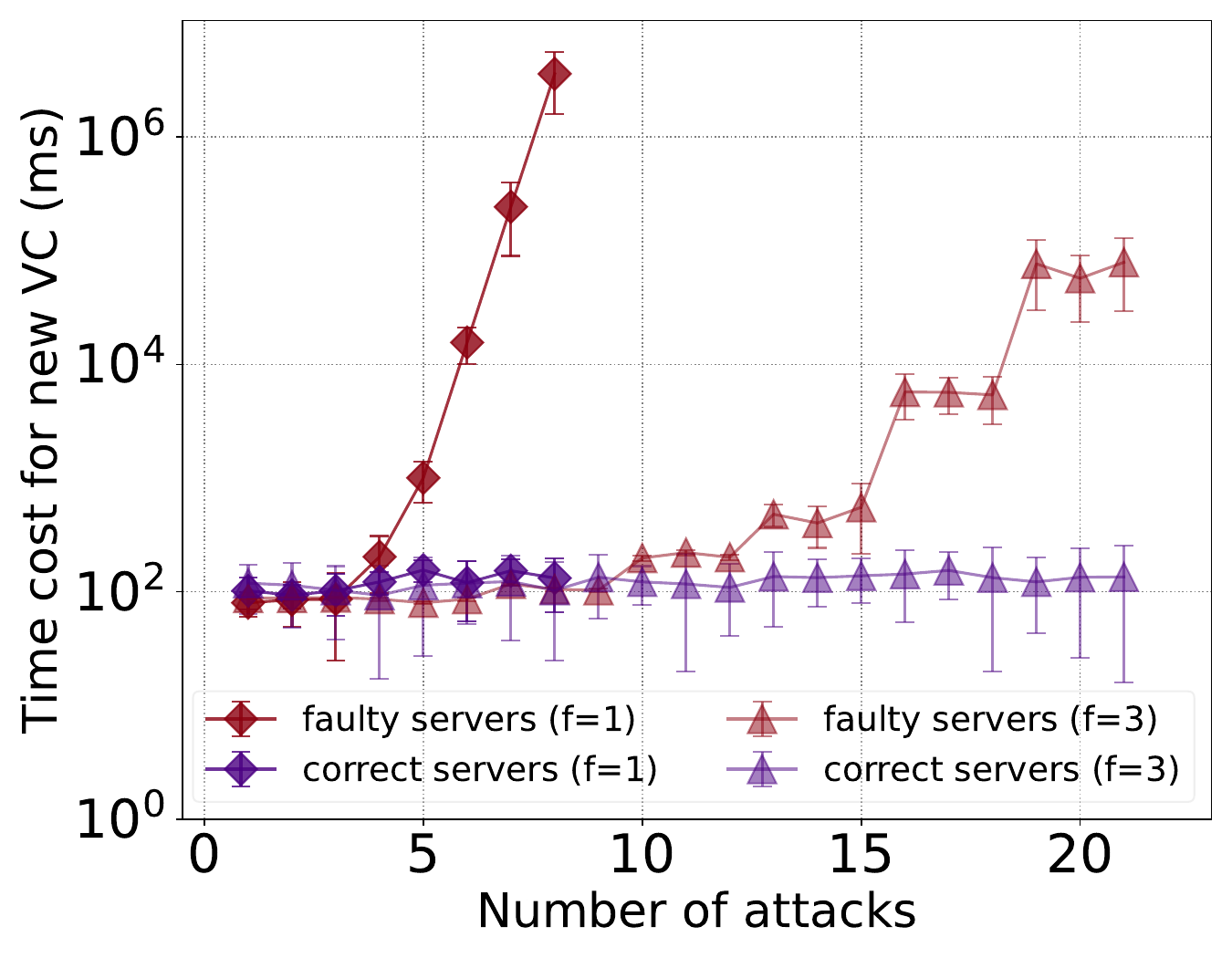}
    \caption{Time costs to start a view change under attacks.}
    \label{fig:eva:timecost}
\endminipage  \hfill
\minipage{0.24\textwidth}
    \centering
    \includegraphics[width=0.97\linewidth]{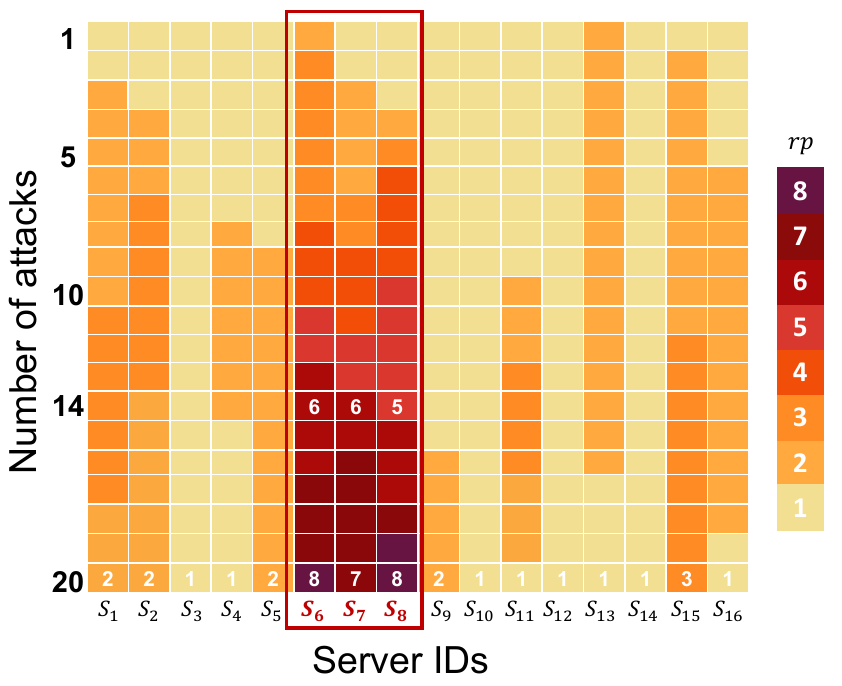}
    \caption{The change of server $rp$ under $f{=}3$ in Fig.~\ref{fig:eva:timecost}.}
    \label{fig:eva:heatmap}
\endminipage \hfill
\minipage{0.24\textwidth}
    \centering
    \includegraphics[width=0.97\linewidth]{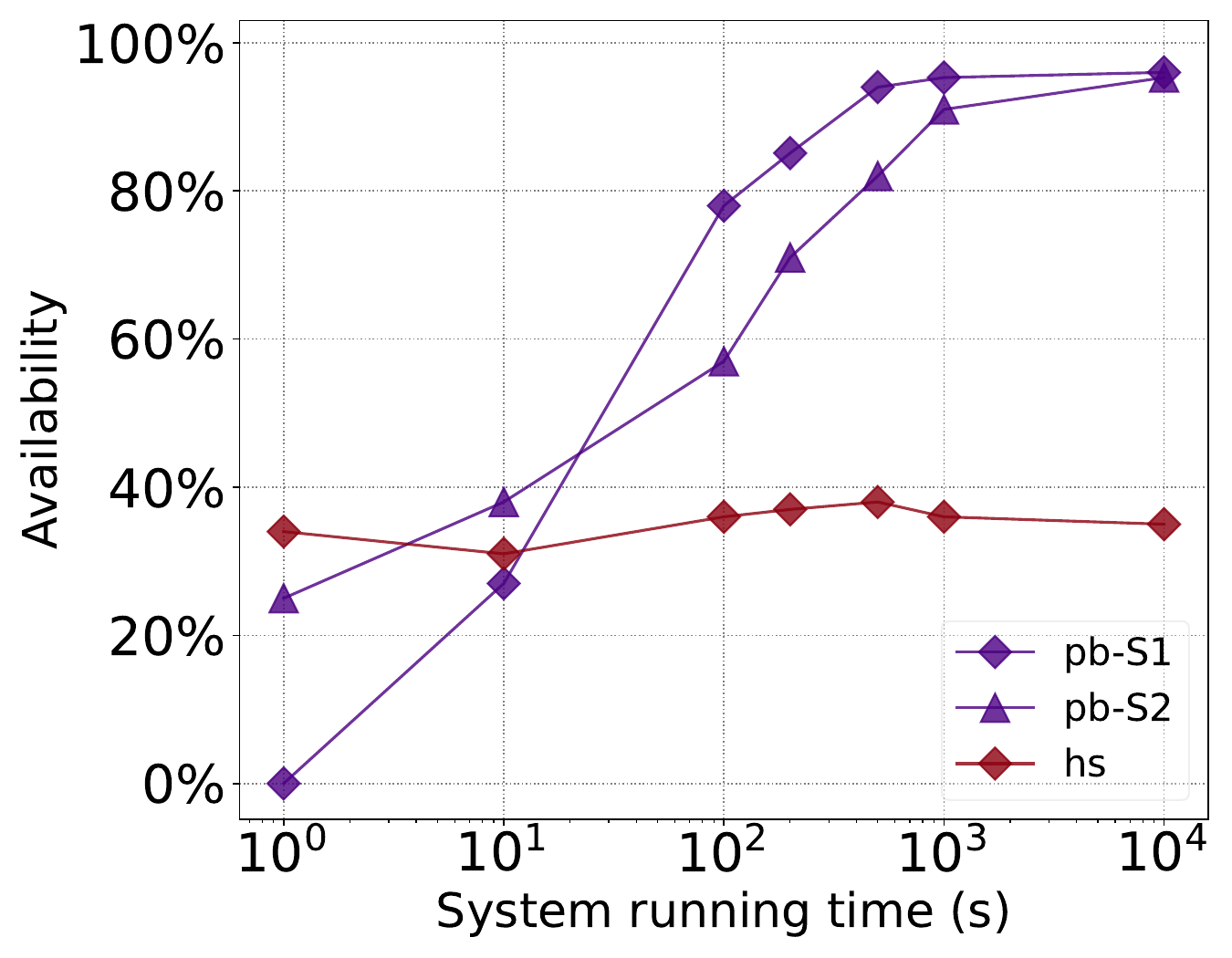}
    \caption{Availability under different types of attacks.}
    \label{fig:eva:availability}
\endminipage
\end{figure*}

In contrast, \algo's throughput was nearly unaffected. Interestingly, with more quiet servers (\ref{byz:quiet}), its throughput saw an increase (e.g., throughput had a $14\%$ from $80,418$ TPS under $f{=}0$ to $91,765$ TPS under $f{=}1$ at \texttt{vb\_r10\_{quiet}}). Since $f$ quiet servers do not consume network bandwidth, more transactions can be exchanged by correct servers.

\textbf{Performance under \ref{byz:vcattacks}+\ref{byz:quiet} and \ref{byz:vcattacks}+\ref{byz:equivocation}.} 
We also evaluated the behavior of repeated view change attacks, which is the most detrimental to \algo's active VC protocol. We arbitrarily chose $f{=}1$ when $n{=}4$ and $f{=}1,3,5$ when $n{=}16$ faulty servers to perform the two combined attacks: when a faulty server becomes a leader, it becomes quiet (\ref{byz:vcattacks}+\ref{byz:quiet}) or perform equivocation (\ref{byz:vcattacks}+~\ref{byz:equivocation}). We allow faulty servers to collude to launch attacks when $f{>}1$ by performing joint computation and sharing logs.

As shown in Figure~\ref{fig:eva:byzfaults}, HotStuff was hit with a similar sustained drop in throughput as in Figure~\ref{fig:eva:quietandequiv}. Its throughput saw a drop of $69\%$ under $n{=}4$, from $32,234$ TPS ($f{=}0$) to $10,051$ TPS ($f{=}1$) at \texttt{hs\_r10\_{quiet}}. Since HotStuff follows a predefined schedule to rotate leadership, faulty servers cannot be selected when they are not scheduled despite they sent view change requests. Compared to the result in Figure~\ref{fig:eva:quietandequiv}, throughput drops slightly higher because of the erroneous messages faulty servers kept sending. 

\algo witnessed a moderate drop in overall throughput. Under $n{=}4$, its throughput dropped by $24\%$, from $80,418$ TPS ($f{=}0$) to $61,208$ TPS ($f{=}1$) at \texttt{vb\_r10\_{quiet}}. In addition, \algo showed an improving performance over the experiment, as faulty servers are constantly penalized after launching attacks without replication. When faulty servers struggled to launch new attacks, \algo gained failure-free views with correct leaders conducting replication. Figure~\ref{fig:eva:byz:vbimproving} shows the trend of the improving throughput: at the beginning of the experiment, \algo suffered from the attacks and could not make progress in replication (in the first $10$s). However, after being repeatedly penalized by the reputation mechanism, faulty servers were quickly suppressed in view changes while correct servers regained leadership and resumed normal operation (staring from $100$s). At time $1000$s, \algo's throughput recovered to $87\%$ of its throughput under normal operation ($f=0$).

We show the time costs for faulty servers launching repeated VC attacks (\ref{byz:vcattacks}+\ref{byz:quiet}). Our implementation uses SHA-256 as the hashing algorithm; thus, the probability of finding a hash that has a prefix of $rp$ leading $0$s is as follows.

\noindent\begin{minipage}{\linewidth}
\centering
$Pr(rp) = \frac{2^{256-8rp}}{2^{256}} = 2^{-8rp}$ 
\end{minipage}
$Pr(rp)$ requires exponentially increasing computation to find a required hash result, which resulted in the skyrocketing time cost for attackers (shown in Figure~\ref{fig:eva:timecost}).

We show the change of server $rp$s in Figure~\ref{fig:eva:heatmap} throughout the attacks in Figure~\ref{fig:eva:timecost} when $f{=}3$, where $S_6$, $S_7$, and $S_8$ are the three faulty servers. After faulty servers' $rp$ exceeded $5$, they began to struggle to perform the required hash computation (Figure~\ref{fig:eva:timecost}) and cannot prevent a correct leader from conducting replication. At this time, correct servers start to regain leadership and apply compensation with reduced $rp$ (from the $14$th attack in Figure~\ref{fig:eva:heatmap}). When the $rp$ of faulty servers increased to $8$, they were unable to launch new attacks and could not become a leader in future view changes.

\textbf{Availability.}
In addition, faulty servers in \algo can ``smartly'' launch repeated VC attacks: they can calculate their $rp$ and launch attacks only when they can get compensated. We name this strategy as \texttt{S2} and the previous strategy, where faulty servers launch attacks whenever they are not the leader, as \texttt{S1}. We kept both \algo and HotStuff ($f{=}3$) running for $10^4$s and reports their availability in Figure~\ref{fig:eva:availability}. To pursue \texttt{S2}, faulty servers must temporarily behave correctly and allow for replication, thereby giving \algo failure-free time. With more transactions replicated, $ti$ in Eq.~\ref{eq:delta-tx} keeps increasing, and faulty servers must behave correctly increasingly longer to continuously get  compensated. Overall, \algo exhibits a significantly higher availability when faulty servers are penalized.

\vspace{-1em}
\subsection{Summary of results}
\label{sec:eva:summary}
The evaluation results show that \wct<1> \algo achieves high performance in replication in terms of throughput and latency. \wct<2> \algo's performance is unaffected by faulty servers being quiet or equivocating, remaining at a high throughput. \wct<3> Under repeated VC attacks, \algo's reputation mechanism quickly suppresses faulty servers with improving performance and availability over time.

\section{Related work}
\label{sec:related-work}

Consensus algorithms provide safety and liveness for state machine replication (SMR)~\cite{schneider1990state} under different failure assumptions. With increasing software scales, Byzantine failures are becoming more common due to hardware glitches, operator errors, and worldwide anonymous collaboration, especially in blockchains where participants may intentionally break the protocol to gain more profit~\cite{zhang2020byzantine, daian2020flash, huang2021rich, lewis2014flash}.

\textbf{Leader-based BFT algorithms} have been favored by permissioned blockchains, such as HyperLedger Fabric~\cite{androulaki2018hyperledger} and Diem~\cite{diem}. After PBFT~\cite{castro1999practical} pioneered a practical BFT solution with an $O(n^2)$ message complexity using public-key signatures, numerous approaches have been proposed for optimizations in terms of replication and view changes.

\textbf{Replication} optimizations have been focusing on various aspects. They use speculative decisions and reduce workloads for a single leader~\cite{kotla2007zyzzyva, duan2014hbft, gunn2019making}, 
develop high performance implementations~\cite{el2019blockchaindb, bessani2014state, guerraoui2010next, buchman2016tendermint, sousa2018byzantine, peng2020falcondb, satija2020blockene, ngo2020tolerating}, use sharding mechanisms to improve throughput~\cite{amiri2021sharper},
reduce messaging costs~\cite{song2008bosco, yang2021dispersedledger, distler2011increasing, martin2006fast, liu2016xft}, 
offer confidentiality protection dealing with secret sharing~\cite{vassantlal2022cobra},
apply accountability for individual participants~\cite{civit2021polygraph, shamis2022ia, neu2021ebb}, limit faulty behavior using trusted hardware~\cite{behl2017hybrids, chun2007attested, kapitza2012cheapbft, levin2009trinc}, and utilize threshold signatures~\cite{libert2016born, shoup2000practical} to achieve linear message complexity~\cite{gueta2019sbft, yin2019hotstuff, zhang2021prosecutor}.

In addition, DAG-based protocols have achieved high performance by separating transaction distribution from consensus~\cite{danezis2022narwhal, spiegelman2022bullshark, keidar2021all}. \algo's view change protocol can also be applied in DAG for efficiently selecting leaders as well as in transaction pipelining~\cite{ford2019threshold, baird2016swirlds}.

\textbf{View changes} detect leader failures and move the system to new views~\cite{aiyer2005bar}. PBFT~\cite{castro1999practical} developed a passive view-change mechanism where servers follow a predefined leader schedule to rotate leaders. Because of its simplicity, this mechanism has been adopted by numerous BFT algorithms~\cite{kotla2007zyzzyva, duan2014hbft, gunn2019making, gueta2019sbft, yin2019hotstuff, clement2009making, abraham2020sync, distler2011increasing, aublin2013rbft, diemConsensus, el2019blockchaindb, bessani2014state, guerraoui2010next, buchman2016tendermint, gupta2020resilientdb}. 
Aardvark~\cite{clement2009making} imposes regular view changes when a leader slows down by a certain threshold, and HotStuff~\cite{yin2019hotstuff} suggests that views be rotated for each request. However, the frequent passive view changes will result in frequent faulty new leaders.

\textbf{Reputation approaches with history.} Learning from the past to predict the future is a common design philosophy in Computer Science, such as the multi-level feedback queue in CPU scheduling~\cite{corbato1962experimental, mcdougall2006solaris}, hardware branch predictors~\cite{mcfarling1993combining, young1994improving}, and caching algorithms~\cite{chrobak1999lru, o1993lru}. Reputation-based BFT algorithms also incorporate historic information and adaptively predict correctness. For example. DiemBFT~\cite{diemConsensus} calculates reputation by tracking active servers' total log length. However, DiemBFT still uses a passive view-change protocol; it restricts the use of reputation only when correct servers are rotated as leaders. In contrast, \algo takes a more proactive approach, fully utilizing reputation with an active view-change protocol to enhance system efficiency and robustness against failures.

\textbf{Leaderless BFT algorithms} do not use a designated leader to conduct replication, thereby mitigating the problem of single points of failures and single server bottlenecks~\cite{lamport2011brief, miller2016honey, crain2018dbft, duan2018beat, suri2021basil}. Without a leader, leaderless BFT algorithms often utilize binary Byzantine agreement~\cite{mostefaoui2014signature} to jointly form quorum certificates~\cite{ben1994asynchronous}. Compared with leader-based BFT algorithms, they are more robust and avoid leadership changes, but often suffer from high message and time costs for conflict resolutions, even after applying erasure coding (e.g., AVID broadcast~\cite{cachin2005asynchronous}).

\section{Conclusions}
This paper introduces \algo, a leader-based BFT consensus algorithm that enables active view changes with reputation mechanisms. The reputation mechanism learns from a server's history and ranks the server's correctness with a reputation penalty. The active view-change protocol allows servers to proactively campaign for leadership by performing reputation-determined work. Consequently, servers with good reputations are more likely to be elected as new leaders than servers with bad reputations. 
Our evaluation results show that \algo is robust and efficient. It achieves $5\times$ higher throughput than its best-performing baseline under normal operation.  It exhibits robustness under failures: while its baseline suffered from a $69\%$ drop in throughput under a variety of Byzantine failures, \algo witnessed only a $24\%$ drop with a vigorous recovery.

\bibliographystyle{ACM-Reference-Format}
\bibliography{ref}

\appendix
\section{Correctness Argument}
\label{sec:ap:correctness}
In this section, we show the correctness argument of \algo's view-change protocol and prove its safety and liveness.
We continue to use the partition of servers as in \S\ref{sec:correctnes}, where $n=3f+1$ servers are divided into three sets: $\mathcal{S}_1$ ($|\mathcal{S}_1|=f+1$) and $\mathcal{S}_2$ ($|\mathcal{S}_2|=f$) are non-faulty, and $\mathcal{S}_f$ ($|\mathcal{S}_f|=f$) are faulty servers.

\subsection{View change correctness} \label{ap:vc-correctness}
We first show the correctness argument of the client interaction of the view-change protocol that attains two key correctness properties: 

\begin{enumerate}
    \item Under a correct leader, no view change will be initiated (leadership robustness).
    \item Under a faulty leader, a view change must be initiated when the faulty leader cannot achieve consensus for client requests (leadership completeness).
\end{enumerate}

To prove the above two key properties, we first define two types of view changes as follows.

\begin{definition}[Unnecessary view changes] \label{def:unnecessaryvc}
Under a non-faulty leader, any view change initiated by other non-faulty servers is an unnecessary view change.
\end{definition}

\begin{definition}[Necessary view changes]
\label{def:necessaryvc}
Under a faulty leader, a view change initiated by a non-faulty server is necessary.
\end{definition}

\algo requires a non-faulty client to broadcast its complaint to all servers. As such, all non-faulty servers ($\mathcal{S}_1 \cup \mathcal{S}_2$ where $|\mathcal{S}_1 \cup \mathcal{S}_2|=2f+1$) are able to receive the complaint (illustrated in Figure~\ref{fig:ap:1:vc:necessity}a) and start the procedure of handling a client complaint (Line~\ref{aglo:af:procedurecompt}). 
Next, we show that faulty clients cannot trigger an unnecessary view change (Definition.~\ref{def:unnecessaryvc}) with or without faulty servers.

\begin{lemma} \label{ap:1:novc-1}
Under a non-faulty leader, faulty clients and non-faulty servers cannot trigger a view change.
\end{lemma}

\begin{proof}
Lemma~\ref{ap:1:novc-1} is straightforward. A faulty client can behave in one of two ways: \wct<1> it sends no complaint, or \wct<2> it sends its complaint to at least one server (illustrated in Figure~\ref{fig:ap:1:vc:necessity}b). 

Scenario \wct<1> simply does not affect our system, which can be disregarded.
In Scenario \wct<2>, when a non-faulty server receives a complaint, it relays it to the leader (Line~\ref{algo:af:relay}). Then, the leader will achieve the consensus for the transaction piggybacked in the complaint. Note that since our failure assumption does not assume DDOS attacks, non-faulty servers are able to handle every proposed request; e.g., if a faulty client sends different transactions to different servers, non-faulty servers will relay every complaint, and the leader will receive them in time. Therefore, all non-faulty servers will terminate this procedure (Line~\ref{fol:compt-committed}) with no view change triggered.
\end{proof}

Next, we show that colluding faulty clients and faulty servers cannot trigger a view change when the current leader is non-faulty.

\begin{figure}[t]
    \centering
    \includegraphics[width=0.85\linewidth]{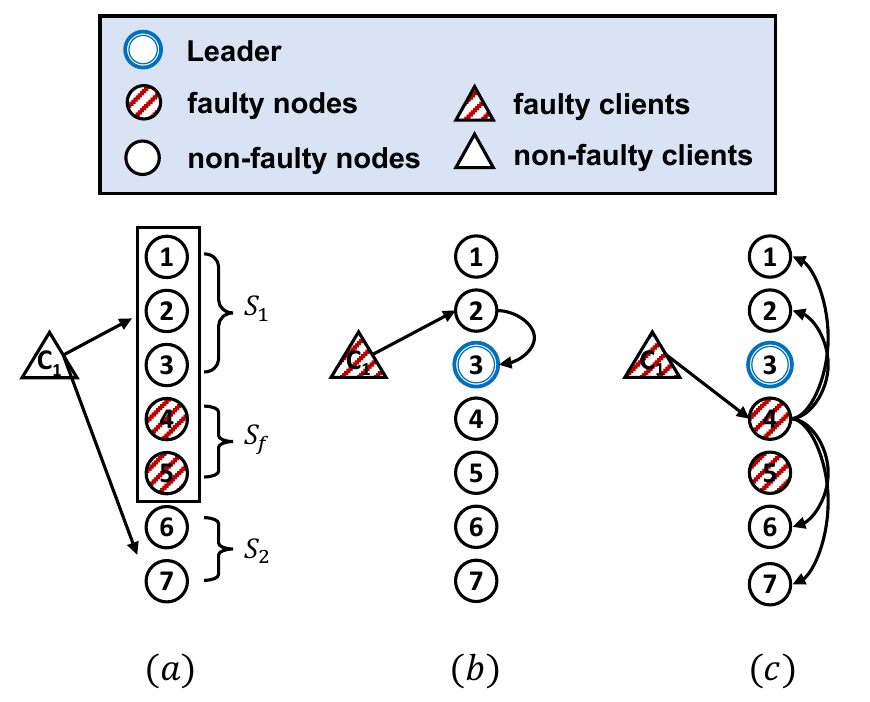}
    \caption{Leadership robustness analysis in \algo's active view-change protocol.}
    \label{fig:ap:1:vc:necessity}
\end{figure}

\begin{lemma} \label{ap:1:novc-2}
Under a non-faulty leader, faulty clients and faulty servers cannot trigger a view change.
\end{lemma}

\begin{proof}
A faulty client in Lemma~\ref{ap:1:novc-2} can behave in one of two ways: \wct<1> it sends its complaint to at least a non-faulty server, and \wct<2> it does not send its complaint to any non-faulty servers. In Scenario \wct<1>, with Lemma~\ref{ap:1:novc-1}, any non-faulty server ($\forall S_i \in \mathcal{S}_1 \cup \mathcal{S}_2$) that receives a complaint will terminate the procedure under a correct leader. In Scenario \wct<2>, the procedure will not be invoked on non-faulty servers.

In the worst case, all faulty servers collude and try to invoke a view change. To make a successful candidate, they have to construct a \texttt{conf\_QC} of size $f+1$ (illustrated in Figure~\ref{fig:ap:1:vc:necessity}c). However, since all non-faulty servers will terminate the procedure, they cannot collect a \textsc{ReVC} from a non-faulty server; i.e., no server in $\forall S_i \in \mathcal{S}_1 \cup \mathcal{S}_2$ will be included in \texttt{conf\_QC}. Consequently, even if a faulty server becomes a candidate, non-faulty servers will not vote for it according to \ref{criterion:2}, and thus a faulty candidate cannot be elected.

Therefore, under a non-faulty leader, no view change will be triggered under $f$ faulty servers and unlimited faulty clients.
\end{proof}

\begin{theorem}[Leadership robustness] \label{ap:leadershiprobustness}
In any given view, under a non-faulty leader, no view change will be initiated.
\end{theorem}

\begin{proof}
Without failures, \algo operates under the replication protocol, and the view-change protocol will not be invoked. With faulty clients and servers, Lemma~\ref{ap:1:novc-1} and~\ref{ap:1:novc-2} have shown that no view change will be invoked under either condition. Therefore, under a non-faulty leader, no view change will be invoked, which proves this theorem.
\end{proof}

Theorem~\ref{ap:leadershiprobustness} is critical for system availability. It shows that \textbf{\algo's active view change protocol will have a stable view under a correct leader regardless of the behavior of faulty clients, faulty servers, and their collusion}.

In addition, with Theorem~\ref{ap:leadershiprobustness}, faulty servers can only intervene in the view-change process when the current leader becomes faulty or a view change is invoked by policy-defined criteria, such as timing policies and throughput-threshold policies (discussed in \S\ref{sec:asfollower}). Next, we show that the interference of faulty servers cannot prevent view changes from being initiated (Definition~\ref{def:necessaryvc}).

\begin{lemma} \label{ap:faultyleaderprevent}
A faulty leader cannot prevent a necessary view change for a higher view.
\end{lemma}

\begin{proof}

When a faulty leader stops committing a non-faulty client's transaction, all non-faulty servers ($\mathcal{S}_1 \cup \mathcal{S}_2$) will receive a complaint from the client. Then, at least a non-faulty server $S_i$ ($S_i \in \mathcal{S}_1 \cup \mathcal{S}_2$) will broadcast a \textsc{ConfVC} message. In this case, $S_i$ can receive at least $f+1$ \textsc{ReCV} replies from $\mathcal{S}_1 \cup \mathcal{S}_2$ and construct a \texttt{conf\_VC} (Line~\ref{algo:af:confqc}) regardless of servers in $\mathcal{S}_f$, starting a new view change with an incremented view. Thus, a faulty leader cannot prevent a necessary view change for a higher view.
\end{proof}

With Lemma~\ref{ap:faultyleaderprevent}, a faulty leader cannot prevent a necessary view change from being initiated by non-faulty servers. However, due to the nature of active view changes, other faulty servers can compete with non-faulty servers in the initiated view change; they may win the election and repossess the leadership. Next, we show the completeness of leadership in \algo; that is, faulty servers cannot indefinitely prevent the election of a non-faulty leader.

\begin{theorem}[Leadership completeness] \label{theorem:leadercompleteness}
In a given view $V$, faulty servers cannot indefinitely prevent a non-faulty leader from being elected in a higher view $V'$ ($V'>V$).
\end{theorem}

\begin{proof}
Faulty servers can delay the appearance of a non-faulty leader by repossessing leadership. 
With Lemma~\ref{lemma:repossess} in~\S\ref{sec:correctnes}, faulty servers cannot repossess leadership indefinitely without making progress in replication. Thus, faulty servers can behave in one of two ways: \wct<1> they launch attacks with the rise of their reputation penalties, or \wct<2> they launch attacks only when they can remain their reputation penalties unchanged by receiving compensation.

Pursuing \wct<1>, after faulty servers exhaust their computation capability, they can no longer be elected as future leaders. Pursuing \wct<2>, in order to get compensated, faulty servers must temporarily give up leadership to non-faulty servers ($\delta_{vc}$) or behave temporarily correctly in replication ($\delta_{tx}$). Therefore, in both ways, faulty servers cannot indefinitely prevent a non-faulty leader from being elected in a higher view.
\end{proof}

\subsection{Liveness}
Next, we prove liveness. We first prove the three properties of \algo's view-change protocol (discussed in~\S\ref{sec:algo:active-vc}).

\begin{lemma} [Property~\ref{vc:p1}] \label{ap:correctness:oneleader}
At most one leader can be elected in a given view.
\end{lemma}
\begin{proof}
We prove this Lemma by contradiction. We claim that there are two legitimate leaders $S_i$ and $S_j$ in a given view $V$. For this claim to be true, $S_i$ and $S_j$ must both have constructed their \texttt{vc\_QC}s of size $2f+1$ when they were candidates, denoted by $\texttt{vc\_QC}_{S_i}$ and $\texttt{vc\_QC}_{S_j}$, respectively. 

In the worst case, $\texttt{vc\_QC}_{S_i}$ is constructed by $\mathcal{S}_1 \cup \mathcal{S}_f$. Since each server votes only once in a view (\ref{criterion:1}), $\mathcal{S}_1$ will not vote for another server in view $V$. However, $\texttt{vc\_QC}_{S_j}$ must also have a size of $2f+1$; it can be composed of $\mathcal{S}_f \cup \mathcal{S}_2$, as $\mathcal{S}_f$ are faulty servers. Nevertheless, since $|\mathcal{S}_f \cup \mathcal{S}_2| = 2f$, it must contain at least one server $S_k$ such that $S_k \in \mathcal{S}_1$. However, no server in $\mathcal{S}_1$ will vote again in view $V$; thus, $\texttt{vc\_QC}_{S_j}$ cannot be formed, which contradicts our claim and proves this Lemma.
\end{proof}

\begin{figure}[t]
    \centering
    \includegraphics[width=0.99\linewidth]{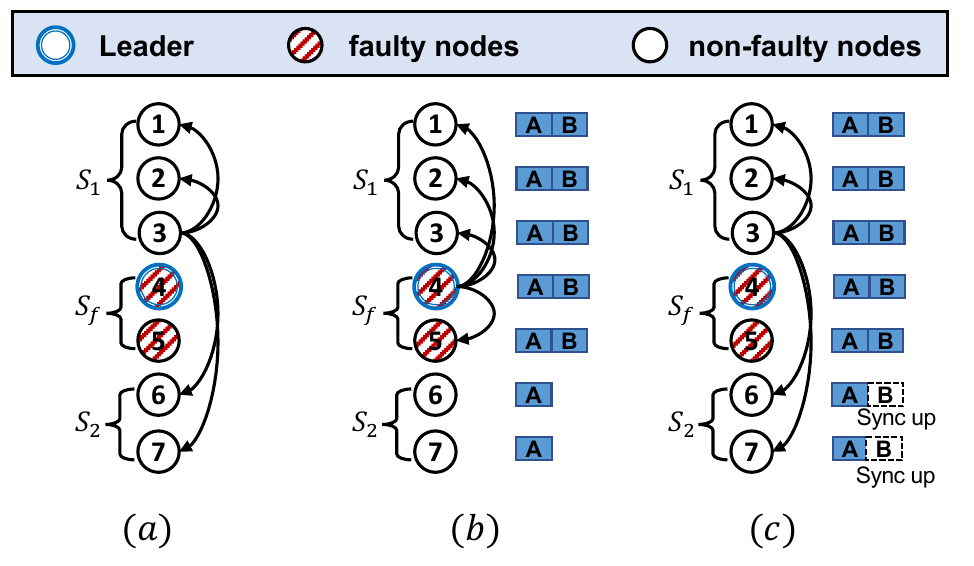}
    \caption{Examples of replication and view changes under a faulty leader.}
    \label{fig:ap:1:vc:properties}
\end{figure}

\begin{lemma} [Property~\ref{vc:property:or}]
An elected non-faulty leader has the most up-to-date replication. \label{lemma:p2}
\end{lemma}

\begin{proof}
The replication protocol (discussed in \S\ref{sec:algo:replication}) requires that a \texttt{txBlock} be committed with a \texttt{commit\_QC} of size $2f+1$. In the worst case, under a faulty leader, the \texttt{commit\_QC} is formed by $\mathcal{S}_1 \cup \mathcal{S}_f$ (illustrated in Figure~\ref{fig:ap:1:vc:properties}b). Similarly, the view-change protocol requires that a \texttt{vcBlock} be committed with a \texttt{vc\_QC} of size $2f+1$. Therefore, at least $f+1$ non-faulty servers have the most up-to-date logs, including \texttt{txBlock}s in replication and \texttt{vcBlock}s in view changes.
According to \ref{criterion:3}, a non-faulty leader will be elected at least from the $f+1$ most up-to-date servers, which proves this Lemma.
\end{proof}

\begin{lemma} [Property~\ref{vc:property:verify}] \label{lemma:property:verify}
An elected leader's reputation penalty and the correspondingly performed computation can be verified by all non-faulty servers. 
\end{lemma}

\begin{proof}
Since \texttt{vcBlock}s are the result of view-change consensus, they are replicated among at least $2f+1$ servers. For the sake of simplicity, we assume that all up-to-date non-faulty servers are in $\mathcal{S}_1$ and stale non-faulty servers are in $\mathcal{S}_2$. In the worst case, \texttt{vcBlock}s are replicated among $\mathcal{S}_1 \cup \mathcal{S}_f$.

From Lemma~\ref{lemma:p2}, a leader is elected among up-to-date servers (i.e., $\mathcal{S}_1$). When servers in $\mathcal{S}_2$ receive a \textsc{VoteCP} from a candidate from $\mathcal{S}_1$ (illustrated in Figure~\ref{fig:ap:1:vc:properties}c), they can verify any more advanced \texttt{txBlock}s and  \texttt{vcBlock}s by checking their $QC$s. Thus, they will sync to up-to-date (Line~\ref{algo:af:syncvc} to~\ref{algo:af:synctx}), obtaining logs as least as up-to-date as the candidate. 
After the sync up, $\mathcal{S}_2$ invokes Algo.~\textsc{CalcRP} using the same input as the candidate. Therefore, the candidate's $rp$ can be reproduced, which can be used to verify its corresponding hash computation result. 
\end{proof}

We have shown that \algo's view-change protocol guarantees the election of an up-to-date leader. Next, we show that it also guarantees that stale servers will not be penalized in unsuccessful elections.

\begin{lemma} \label{lemma:unsuccessfulattempts}
The reputation penalties of non-faulty but stale servers will not be increased in view changes.
\end{lemma}

\begin{proof}
Although a stale server's leader election will not be successful, the stale server can still invoke the view-change protocol and transition to the candidate state. It will not receive sufficient votes because up-to-date servers in $\mathcal{S}_1$ will never vote for it. In this case, its calculated reputation penalty will not be recorded in the \texttt{vcBlock} of the new view. Note that in each view change, only the elected leader's reputation penalty and compensation index are updated (discussed in \S\ref{sec:asleader}). Therefore, unsuccessful attempts of leader election will not change a server's reputation penalty.
\end{proof}

Now we show the proof of liveness; we repeat the theorem of liveness below:

\begin{theorem}[Liveness] (Same as Theorem~\ref{theorem:liveness:paper})
\label{theorem:liveness:appendix}
After GST, a non-faulty server eventually commits a proposed client request.
\end{theorem}

\begin{proof}
At any given time, leadership is in one of the following two conditions: \wct<1> leadership is controlled by $f$ faulty servers, or \wct<2> leadership is released by $f$ faulty servers.

In \wct<1>, with Lemma~\ref{lemma:repossess}, faulty leaders must at some point start to conduct replication. Otherwise, they cannot  control the leadership indefinitely. When they start to conduct replication, they become temporary non-faulty leaders.

In \wct<2>, with Lemma~\ref{lemma:oneleader}, a leader (e.g., $S_i$) will eventually be elected from up-to-date and non-faulty servers (i.e., $S_i \in \mathcal{S}_1$). 
With Lemma~\ref{lemma:unsuccessfulattempts}, the reputation penalties of stable servers do not increase in view changes. Note that in case of rising reputation penalties incurred by GST, all non-faulty servers can apply refresh penalties introduced in \S\ref{sec:optimizations}. 

In addition, we show that all non-faulty servers are able to move to a new view. Assume a server $S_i$ in a view $V$, which has one of three possible scenarios in a view change: \wct<1> $S_i$ initiates a leader election campaign for a higher view $V'$ ($V'>V$) and is elected as the new leader; \wct<2> $S_i$ initiates a leader election campaign for a higher view $V'$ ($V'>V$) but is not elected as the new leader; and \wct<3> $S_i$ does not initiate a leader election campaign and remained as a follower in view $V$.

Scenario \wct<1> is straightforward. When $S_i$ is elected, $S_i$ moves to the new view it initiated. In Scenario \wct<2>, if $S_i$ did not win an election, it can be in the redeemer state or the candidate state. In both states, its operating view is $V$ (from the current \texttt{vcBlock} of view $V$), and the view it is campaigning for is $V'$ ($V'>V$). Once $S_i$ receives a legit \texttt{vcBlock} of view $V^*$ ($V^*>V$), it aborts its campaign activity for view $V'$, moving to view $V^*$ by transitioning back to the follower state in accordance with the procedure of receiving a new \texttt{vcBlock} defined in \S\ref{sec:asleader}.
In Scenario \wct<3>, $S_i$ simply follows the same procedure moving to the new view when it receives a \texttt{vcBlock} of a higher view.

To conclude, after GST, all servers are able to move to a new view, and a non-faulty leader is eventually elected in the new view. Therefore, a client request will eventually be committed by all non-faulty servers; i.e., \algo ensures that a client eventually receives replies to its request after GST.
\end{proof}

\subsection{Safety}

After a leader is elected in a view, \algo uses a standard two-phase replication protocol to conduct consensus for transactions proposed by clients. We now prove that \algo ensures safety.

\begin{theorem}[Safety] (Same as Theorem~\ref{theorem:safety})
Non-faulty servers do not decide on conflicting blocks. That is, non-faulty servers do not commit two \texttt{txBlock}s at the same sequence number $n$.
\end{theorem}

\begin{proof}
With Lemma~\ref{lemma:oneleader}, \algo ensures that each view has at most one leader. When a view has a non-faulty leader, the replication protocol is invoked to conduct consensus for transactions proposed by clients and produces the consensus result as a \texttt{txBlock} with a unique sequence number $n$. 
Note that \algo does not allow the consensus process of a \texttt{txBlock} to operate across views, as servers never respond to a message from a lower view (discussed in \S\ref{sec:algo:replication}). The \texttt{ordering\_QC} and \texttt{commit\_QC} must be constructed in the same view. We claim that there are \texttt{txBlock} and $\texttt{txBlock}_{\diamond}$, both committed with sequence number $n$. 

In this case, \texttt{commit\_QC} and $\texttt{commit\_QC}_{\diamond}$ are both signed by $2f+1$ servers. Say \texttt{commit\_QC} is signed by servers in $\mathcal{S}_1 \cup \mathcal{S}_f$. Then, servers in $\mathcal{S}_1$ cannot sign $\texttt{commit\_QC}_{\diamond}$ with $n$. Although faulty servers in $\mathcal{S}_f$ can double commit, $\texttt{commit\_QC}_{\diamond}$ can only find servers in $\mathcal{S}_f \cup \mathcal{S}_2$ ($|\mathcal{S}_f|{+}|\mathcal{S}_2|{=}2f$) to sign it, which is not sufficient to form a $QC$ of size $2f+1$. Therefore, $\texttt{commit\_QC}_{\diamond}$ cannot be formed, which contradicts our claim.

In addition, with Lemma~\ref{lemma:p2} and Theorem~\ref{theorem:leadercompleteness}, a non-faulty server that has the most up-to-date log will be elected as a leader for normal operation. Thus, a non-faulty leader is always aware of the highest sequence number and will not reassign a used sequence number for a \texttt{txBlock}.

Therefore, the combination of \algo's view-change protocol and the standard two-phase replication protocol ensures safety, with no non-faulty servers deciding on conflicting blocks.
\end{proof}

\subsection{Leadership fairness}

Since the passive view-change protocol rotates leadership according to a predefined leader schedule, it intrinsically achieves leadership fairness as each server becomes a leader once in a circle of rotations. However, its leadership fairness is shared among all servers including faulty ones, which can always result in regular faulty views with unavailable leaders, especially under frequent view changes.

In contrast, \algo's active view-change protocol achieves a stronger form of leadership fairness. Since faulty servers are penalized with worsening reputation penalties after showing a pattern of launching attacks, leadership will be eventually shared among non-faulty servers over the long run.

\begin{theorem}[Strong leadership fairness]
\algo eventually achieves leadership fairness among all non-faulty servers.
\end{theorem}

\begin{proof}
With Lemma~\ref{lemma:repossess}, faulty servers may \wct<1> become faulty leaders with increasing reputation penalties or \wct<2> become temporary non-faulty servers and launch attacks when they can get compensated.

In Scenario \wct<1>, after faulty servers exhaust their computation capability, leadership will be campaigned by only non-faulty servers. In Scenario \wct<2>, during the period when faulty servers behave correctly, leadership will also be campaigned by only non-faulty servers. Therefore, \algo eventually achieves leadership fairness among all non-faulty servers.
\end{proof}

\section{Collected questions}
\label{sec:ap:collectedq}
\renewcommand*{\proofname}{Answer}

In this section, we show questions that were collected during presentations, lectures, and conversions from various groups including ECE/CS graduate students, professors, and distributed system developers. Questions are arranged according to their related sections.

\begin{question}[Motivation]
The passive view-change protocol indeed suffers from performance degradation, but the good thing about passive VC is that it can decide on a leader regardless of whether the $f$ failures are crash failures or Byzantine failures. How does \algo perform under a variety of attacks compared to the passive protocol?
\end{question}

\begin{proof}
Compared to the simple passive VC protocol, \algo has a more advanced and sophisticated VC protocol. As shown in the evaluation section, \algo outperforms the passive VC protocol both under crash and Byzantine failures. When it comes to tolerating crash failures, \algo shows a significant advantage. Since \algo allows servers to actively campaign for leadership upon detecting a leader's failure, it never assigns an unavailable or a stale server as a leader. Additionally, the evaluation of quiet attacks (\ref{byz:quiet}), similar to crash failures, demonstrates that \algo remains unaffected while the passive VC protocol is severely impacted (see Figure~\ref{fig:eva:quietandequiv}).

Regarding tolerating Byzantine failures, \algo has the capability to mitigate the impact of arbitrary faults and progressively improve its availability over time. Despite the fact that faulty servers can launch attacks that come with computational costs, \algo may experience a brief period of low availability while increasing faulty servers' reputation penalties. However, \algo surpasses the passive VC protocol as soon as its reputation mechanism responds appropriately to accumulated historical data in view changes and replication (see Figure~\ref{fig:eva:availability}).
\end{proof}

\begin{question}[Motivation]
Why a speculative approach? Can we kick faulty servers out when some servers fail and reconfigure the system?
\end{question}

\begin{proof}
Excluding faulty servers can be a temporary solution to deal with failures, but it does not represent a fault-tolerant approach. The focus of fault tolerance is to ensure that the system continues to function correctly even in the presence of failures.

Furthermore, in the context of Byzantine fault tolerance, distinguishing between benign and malicious behavior can be difficult. It is often impossible to determine whether a server is intentionally dropping a request or if the network is responsible for the failure. If we continuously exclude servers every time they exhibit a failure, we may soon find ourselves running out of servers, leading to frequent and manual configuration changes.
\end{proof}

\begin{question}[Reputation mechanism]
What if bad clients collude with faulty leaders and send bad requests to the system to let the faulty leader gain some reputation and in turn let faulty servers enjoy penalty deductions?
\end{question}

\begin{proof}
\algo leaves the judgment of good and bad requests to the applications. As discussed in \S\ref{sec:reputation}, users can define the criteria of useful \texttt{txBlock}s and the impact factor $C_{\delta}$ based on specific use cases. For example, in a financial application, a \texttt{txBlock} can be considered useful if its transactions are worth more than $\$1,000$, while transactions below this amount will not be counted in \texttt{ti} to receive compensation. This strategy can prevent frequent small transactions from impacting the calculation of reputation penalties.

\algo proposes a general and versatile architecture incorporating a behavior-aware reputation mechanism, providing flexibility to its applications. This architecture enables user-defined information to convert behavior into a reputation penalty, which can be tailored to each application's unique requirements.
\end{proof}

\begin{question}[Reputation mechanism]
Will the increasing value of \texttt{ti} in the incremental log responsiveness make it more challenging for servers to receive compensation over time?
\end{question}

\begin{proof}
The criterion of incremental log responsiveness is intended to reward servers that make increasing progress in replication, which prevents faulty servers from receiving compensation for making only limited progress. When faulty servers temporarily pretend to be correct in order to receive compensation, this criterion forces them to keep replicating more transactions after each time they receive compensation (e.g., examples \wct<3> vs. \wct<4> in Figure~\ref{fig:reputation-examples}). This design has resulted in an improvement in availability when faulty servers choose to launch attacks only when they can receive compensation. In the long run, when the reputation penalties of at least $f+1$ non-faulty servers exceed the predefined threshold, the refresh mechanism will reset $rp$ and $ci$ to the initial value for these servers (discussed in~\ref{sec:optimizations}). Consequently, the refresh will ``rejuvenate'' the calculation of $\delta_{tx}$.
\end{proof}

\begin{question}[Reputation mechanism]
The reputation design is interesting. Your current approach seems to only reduce the interference of faulty servers in view changes. Can the reputation mechanism be adapted to also reduce the interference of faulty servers in replication? If so, will this increase the performance even more?
\end{question}

\begin{proof}
The primary focus of \algo is on view changes, as faulty leaders have the most harmful impact on leader-based BFT algorithms. While we have considered the possibility of introducing penalization in replication, we have two major concerns that have hindered us from implementing this feature. Firstly, under a faulty leader, $f$ correct servers can always be blacked out in replication. Thus, it is not possible to judge reputation based on states, as a $QC$ can always be constructed by $f+1$ correct servers and $f$ faulty servers. Secondly, the reputation mechanism is currently only activated during view changes, which does not impose any additional overhead on replication. Penalizing wrongdoing during replication may require additional message passing among servers, which could introduce overhead.

However, we remain open to the idea of introducing penalization in replication in the future, as we continue to explore ways to build up more efficient and more robust fault tolerance algorithms.
\end{proof}

\begin{question}[View changes]
You mentioned that VDF is an alternative way of using PoW to implement the effect of reputation penalties. How would incorporating VDF to implement the effect of reputation penalties change the overall architecture?
\end{question}

\begin{proof}
Changing PoW to VDF will not change \algo's overall architecture. In fact, the reputation mechanism does not need to change at all.
To use VDF, we first change the hash computation process of a redeemer (Line~\ref{algo:ar:powstarts} to~\ref{algo:ar:powends}) to a delay function where the delay time is determined by the reputation penalty. Then, we change the verification of PoW computational work (\ref{criterion:5} in \S\ref{sec:ascandidate}) to the verification of delayed time. 
\end{proof}

\begin{question}[View changes]
What if a faulty candidate colludes with a faulty leader and tricks correct servers by sending them a block that is not the latest one, since stale servers cannot know what the latest transaction block is?
\end{question}

\begin{proof}
This is a possible scenario, but it will not affect the correctness. In replication, each $QC$ has a size of $2f+1$, there must be at least $f+1$ up-to-date and non-faulty servers knowing the latest \texttt{txBlock}. The faulty candidate cannot receive sufficient votes because all up-to-date and non-faulty servers will never vote for them. Our view-change protocol ensures that an elected leader must have the most up-to-date logs (Property~\ref{vc:property:or}). In addition, when stale servers receive a \textsc{CampVC} message from an up-to-date and non-faulty candidate, they will sync to up-to-date, regardless of any tricks from faulty candidates.
\end{proof}

\begin{question}[View changes]
In addition to leadership fairness, how can your algorithm support fairness in handling client requests?
\end{question}

\begin{proof}
\algo's active VC protocol supports the fairness problem of handling client requests by making frequent view changes more efficient and robust. As discussed in the introduction, faulty leaders can unfairly handle client requests. For example, faulty leaders can choose to handle the requests from selected clients first and intentionally slow down the consensus process for targeted clients. Some approaches such as Aardvark~\cite{clement2009making} and Diem~\cite{diem} have proposed approaches to frequently change leadership through view changes in order to mitigate the unfair handling problem. \algo's active VC protocol can be applied to replace the passive VC protocol used in these approaches with enhanced efficiency and robustness, leading to high performance in terms of throughput and latency under frequent view changes.
\end{proof}

\section{Examples}
\label{sec:ap:examples}

In this section, we show the step-by-step calculations that the reputation mechanism converts a server's behavior history into a reputation penalty. 
By walking through these calculations, we aim to provide a clear and comprehensive understanding of how the reputation mechanism operates and the results are used in view changes.

We assume a $4$-server system including servers $S_1$, $S_2$, $S_3$, and $S_4$. The initial view is $V1$ where $rp=1$ and $ci=1$ for each server. The reputation segment of the initial \texttt{vcBlock} of view $V1$ (denoted by \texttt{vcBlock[V1]}) is presented below:

\noindent\begin{minipage}{.49\linewidth}
\scriptsize
\[
\texttt{vcBlock[V1]}.rp = 
    \left\{
    \begin{aligned}
    <ID: 1, rp: 1>\\
    <ID: 2, rp: 1>\\
    <ID: 3, rp: 1>\\
    <ID: 4, rp: 1>
    \end{aligned}
    \right.
    \]
\end{minipage}
\hfill
\begin{minipage}{.49\linewidth}
\scriptsize
    \[
    \texttt{vcBlock[V1]}.ci = 
    \left\{
    \begin{aligned}
    <ID: 1, ci: 1>\\
    <ID: 2, ci: 1>\\
    <ID: 3, ci: 1>\\
    <ID: 4, ci: 1>    
    \end{aligned}
    \right.
    \]
\end{minipage}

~\\

We show how the reputation penalty ($rp$) and compensation index ($ci$) of server $S_1$ are calculated based on its different behavior in view changes and replication.

\begin{figure}[h]
    \centering
    \includegraphics[width=0.65\linewidth]{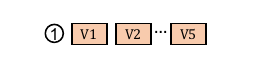}
\end{figure}

The blocks in \wct<1> show that $S_1$ has repeatedly possessed the leadership without making progress in replication. After $V1$, $S_1$ campaigned for view $V2$. It first goes through Eq.~\ref{eq:increase-penalty} (penalization), where $V=1$, $V'=2$, and $rp^{(1)}=1$:

$$rp^{(2)}_{temp}=rp^{(1)}+V'-V=2$$

After this, $S_1$ is not eligible to receive compensation because it has not replicated any transaction, resulting in its $\delta_{tx}=0$.
Thus, $rp^{(2)} = rp^{(2)}_{temp}=2$, and \texttt{vcBlock[V2]} that $S_1$ prepares for the view $V2$ is as follows:

\noindent\begin{minipage}{.49\linewidth}
\scriptsize
\[
\texttt{vcBlock[V2]}.rp = 
    \left\{
    \begin{aligned}
    <ID: 1, rp: 2>\\
    <ID: 2, rp: 1>\\
    <ID: 3, rp: 1>\\
    <ID: 4, rp: 1>
    \end{aligned}
    \right.
    \]
\end{minipage}
\hfill
\begin{minipage}{.49\linewidth}
\scriptsize
    \[
    \texttt{vcBlock[V2]}.ci = 
    \left\{
    \begin{aligned}
    <ID: 1, ci: 1>\\
    <ID: 2, ci: 1>\\
    <ID: 3, ci: 1>\\
    <ID: 4, ci: 1>    
    \end{aligned}
    \right.
    \]
\end{minipage}

$S_1$ repeats this behavior to view $V5$, and \texttt{vcBlock[V5]} is as follows:

\noindent\begin{minipage}{.49\linewidth}
\scriptsize
\[
\texttt{vcBlock[V5]}.rp = 
    \left\{
    \begin{aligned}
    <ID: 1, rp: 5>\\
    <ID: 2, rp: 1>\\
    <ID: 3, rp: 1>\\
    <ID: 4, rp: 1>
    \end{aligned}
    \right.
    \]
\end{minipage}
\hfill
\begin{minipage}{.49\linewidth}
\scriptsize
    \[
    \texttt{vcBlock[V5]}.ci = 
    \left\{
    \begin{aligned}
    <ID: 1, ci: 1>\\
    <ID: 2, ci: 1>\\
    <ID: 3, ci: 1>\\
    <ID: 4, ci: 1>    
    \end{aligned}
    \right.
    \]
\end{minipage}

\begin{figure}[h]
    \centering
    \includegraphics[width=0.8\linewidth]{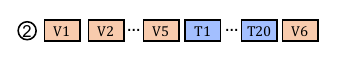}
\end{figure}

In view $V5$, $S_1$ can not afford the increasing $rp$ and decides to temporarily behave like a correct leader. 
The blocks in \wct<2> indicates that $S_1$ has replicated $20$ transactions (i.e., \texttt{txBlock[T1]} to txBlock[T20]). Then, at the end of $V5$, when $S_1$ campaigns for $V6$, it first gets penalized by Eq.~\ref{eq:increase-penalty}, where $V=5$, $V'=6$, and $rp^{(5)}=5$:

$$rp^{(6)}_{temp}=rp^{(5)}+V'-V=6 $$

Then, $S_1$ can receive compensation, as it has replicated $20$ \texttt{txBlock}s ($\texttt{ti}=20$).

$$\delta_{tx} =  \dfrac{\texttt{ti} - \texttt{ci}}{\texttt{ti}} = 1$$

Since $\mathcal{P} = \{1,2,3,4,5\}$, $\mu_{\mathcal{P}}=3$ and $\sigma_{\mathcal{P}}=1.41$; thus,

$$\delta_{vc} = 1- Sigmoid(\dfrac{rp^{(5)}-\mu_{\mathcal{P}}}{\sigma_{\mathcal{P}}}) = 0.19
$$

$$\delta= \delta_{tx} \times \delta_{vc} \times rp^{(6)}_{temp} = 1.14$$

$$rp^{(6)} = rp^{(6)}_{temp} - \floor{\delta} = 5$$

Consequently, the reputation segment of \texttt{vcBlock[V6]} that $S_1$ sends is as follows:

\noindent\begin{minipage}{.49\linewidth}
\scriptsize
\[
\texttt{vcBlock[V6]}.rp = 
    \left\{
    \begin{aligned}
    <ID: 1, rp: 5>\\
    <ID: 2, rp: 1>\\
    <ID: 3, rp: 1>\\
    <ID: 4, rp: 1>
    \end{aligned}
    \right.
    \]
\end{minipage}
\hfill
\begin{minipage}{.49\linewidth}
\scriptsize
    \[
    \texttt{vcBlock[V6]}.ci = 
    \left\{
    \begin{aligned}
    &<ID: 1, ci: 20>\\
    &<ID: 2, ci: 1>\\
    &<ID: 3, ci: 1>\\
    &<ID: 4, ci: 1>    
    \end{aligned}
    \right.
    \]
\end{minipage}

\begin{figure}[h]
    \centering
    \includegraphics[width=\linewidth]{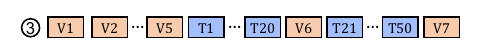}
\end{figure}

In view $V6$, $S_1$ wants to get compensated again, so as shown in \wct<3>, it temporarily behaves correctly again by replicating another $30$ \texttt{txBlock}s (i.e., $50$ \texttt{txBlock}s in total). At the end of $V6$, when $S_1$ campaigns for view $V7$. It first gets penalized by Eq.~\ref{eq:increase-penalty}, where $V=6, V'=7, rp^{(6)}=5$:

$$rp^{(7)}_{temp}=rp^{(6)}+V'-V=6 $$

Then, the calculation moves to compensation. Since $S_1$ has used $20$ \texttt{txBlock}s for the last compensation calculation, its $\texttt{ci}=20$ and $\texttt{ti}=50$.

$$\delta_{tx} =  \dfrac{\texttt{ti} - \texttt{ci}}{\texttt{ti}} = 0.6$$

Since $\mathcal{P} = \{1,2,3,4,5,5\}$, $\mu_{\mathcal{P}}=3.33$ and $\sigma_{\mathcal{P}}=1.49$; thus,

$$\delta_{vc} = 1- Sigmoid(\dfrac{rp^{(5)}-\mu_{\mathcal{P}}}{\sigma_{\mathcal{P}}}) = 0.25$$

$$\delta= \delta_{tx} \times \delta_{vc} \times rp^{(7)}_{temp} = 0.89$$

$$rp^{(7)} = rp^{(7)}_{temp} - \floor{\delta} = 6$$

In \wct<3>, $S_1$ cannot receive any compensation, and its reputation penalty increases to $rp=6$. 
If $S_1$ wants to receive compensation in this view, it must get a higher $\delta_{tx}$, as $\delta_{vc}$ remains unchanged in a view. For example, if $S_1$ replicates $80$ more \texttt{txBlock}s in this view with $\texttt{ti}=100$ (shown in \wct<4>), its $\delta_{tx}$ will result in compensation.

\begin{figure}[h]
    \centering
    \includegraphics[width=\linewidth]{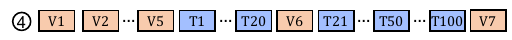}
\end{figure}

$$\delta_{tx} =  \dfrac{\texttt{ti} - \texttt{ci}}{\texttt{ti}} = 0.8$$

In this case, with $\delta_{vc}=0.25$,
$$\delta= \delta_{tx} \times \delta_{vc} \times rp^{(7)}_{temp} = 1.2$$

Thus, $S_1$ will be compensated.
$$rp^{(7)} = rp^{(7)}_{temp} - \floor{\delta} = 5$$

Consequently, the reputation segment of \texttt{vcBlock[V7]} that $S_1$ sends is as follows:

\noindent\begin{minipage}{.49\linewidth}
\scriptsize
\[
\texttt{vcBlock[V7]}.rp = 
    \left\{
    \begin{aligned}
    <ID: 1, rp: 5>\\
    <ID: 2, rp: 1>\\
    <ID: 3, rp: 1>\\
    <ID: 4, rp: 1>
    \end{aligned}
    \right.
    \]
\end{minipage}
\hfill
\begin{minipage}{.49\linewidth}
\scriptsize
    \[
    \texttt{vcBlock[V7]}.ci = 
    \left\{
    \begin{aligned}
    &<ID: 1, ci: 100>\\
    &<ID: 2, ci: 1>\\
    &<ID: 3, ci: 1>\\
    &<ID: 4, ci: 1>    
    \end{aligned}
    \right.
    \]
\end{minipage}

~\\

The comparison of \wct<3> and \wct<4> shows that \textbf{the criterion of replication ($\delta_{tx}$) entices servers to behave correctly by incentivizing a more up-to-date replication log.} In order to continuously receive compensation, $S_1$ must make incrementally growing progress in replication. If $S_1$ is a faulty server and only temporarily behaves correctly to receive compensation, the temporary period increases significantly after each time $S_1$ receives compensation.

In addition to enticing servers to have a more up-to-date replication, the \textbf{reputation mechanism also incentives ``heavily penalized servers'' to give up leadership} and stay as a follower for a while (e.g., until they can receive compensation again).

\begin{figure}[h]
    \centering
    \includegraphics[width=0.95\linewidth]{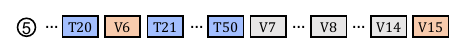}
\end{figure}

For example, in \wct<5>, if $S_1$ gives up leadership and does not campaign for view $V7$, its $rp$ and $ci$ remain unchanged. If $S_1$ operates as a follower through view $V7$ to $V14$ (gray \texttt{vcBlock}s), then at the end of $V14$, its $rp=5$.

$$\mathcal{P} = \lbrace 1, 2, 3, 4, 5, 5, ... , 5 \rbrace \text{ //5 appears 10 times}$$ 

If $S_1$ campaigns for view $V15$, after penalization,
$$rp^{(15)}_{temp}=rp^{(14)}+V'-V=6 $$

Then, it goes through compensation with $\mu_{\mathcal{P}}=4.28$ and $\sigma_{\mathcal{P}}=1.27$.

$$\delta_{tx} =  \dfrac{\texttt{ti} - \texttt{ci}}{\texttt{ti}} = 0.6$$

$$\delta_{vc} = 1- Sigmoid(\dfrac{rp^{(14)}-\mu_{\mathcal{P}}}{\sigma_{\mathcal{P}}}) = 0.36$$

$$\delta= \delta_{tx} \times \delta_{vc} \times rp^{(15)}_{temp} = 1.29$$

Thus, $S_1$ will be compensated by a deduction of $1$ with its $rp$ unchanged.
$$rp^{(15)} = rp^{(15)}_{temp} - \floor{\delta} = 5$$

In addition, if $S_1$ has replicated $400$ \texttt{txBlock}s throughout the $14$ views with its $\texttt{ti}=400$ (as shown in \wct<6>), it will receive higher compensation for the better behavior from both sides.

\begin{figure}[h]
    \centering
    \includegraphics[width=\linewidth]{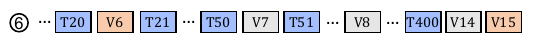}
\end{figure}

$$\delta_{tx} =  \dfrac{\texttt{ti} - \texttt{ci}}{\texttt{ti}} = 0.95$$

$$\delta= \delta_{tx} \times \delta_{vc} \times rp^{(15)}_{temp} = 2.05$$

Then, $S_i$ will be compensated by a deduction of $2$ with its $rp$ decreased to $4$.
$$rp^{(15)} = rp^{(15)}_{temp} - \floor{\delta} = 4$$

In this section, we have demonstrated how the reputation mechanism calculates a server's reputation penalty and compensation index during view changes based on its behavior history through various examples. We have shown how the reputation mechanism incentivizes servers to maintain up-to-date replication and avoid frequent leadership repossession.
\end{document}